\documentclass[12pt]{article}
\usepackage{amsfonts}
\usepackage{amscd}
\usepackage{amsmath}
\usepackage{epsfig}
\usepackage{amsthm}
\usepackage{amssymb}
\usepackage{amsbsy}
\usepackage{latexsym}
\usepackage{eucal}
\usepackage{mathrsfs,bm}
\usepackage{mathtools}
\usepackage{tikz}
\usepackage[colorlinks]{hyperref}
\usepackage{color}

\DeclareMathOperator{\csch}{csch}

 \newtheorem{theorem}{Theorem}[section]
\newtheorem{definition}[theorem]{Definition}
\newtheorem{proposition}[theorem]{Proposition}
\newtheorem{corollary}[theorem]{Corollary}
\newtheorem{remark}[theorem]{Remark}
\newtheorem{lemma}[theorem]{Lemma}

\def\be{\begin{equation}}
\def\ee{\end{equation}}
\def\ben{\begin{displaymath}}
\def\een{\end{displaymath}}
\def\baa{\begin{eqnarray}}
\def\eaa{\end{eqnarray}}

\def\ba{\begin{array}}
\def\ea{\end{array}}
\renewcommand{\leq}{\leqslant}
\renewcommand{\geq}{\geqslant}

\newcommand{\supp}{\operatorname{supp}}

\newcommand{\Tr}{\operatorname{Tr}}
\newcommand{\Res}{\operatorname{Res}}

\makeatletter
\@addtoreset{equation}{section}
\makeatother

\begin{document}
\title {Polyakov-Alvarez type comparison formulas for determinants of Laplacians on Riemann surfaces with conical singularities}

\author{Victor Kalvin \footnote{{\bf E-mail:  vkalvin@gmail.com}}}

\date{}
\maketitle

\begin{abstract}  We present and prove Polyakov-Alvarez type comparison formulas for the determinants of  Frie\-de\-richs  extensions of Laplacians corresponding to  conformally equivalent metrics on a compact Riemann surface with conical singularities. In particular, we  find how the determinants depend on the orders of conical singularities. We also illustrate these general results with  several examples: based on our Polyakov-Alvarez type formulas we  recover  known and obtain new  explicit formulas for determinants of Laplacians on  singular surfaces with and without boundary.  
In one of the examples  we show that  on the metrics of constant curvature on a sphere with two conical singularities and fixed area $4\pi$  the determinant of Friederichs Laplacian is unbounded from above and  attains its  local maximum on the metric of standard round sphere.   In another example  we deduce the famous Aurell-Salomonson formula for the determinant of Friederichs Laplacian on  polyhedra with spherical topology, thus providing the formula with mathematically rigorous proof. \end{abstract}


\noindent{MSC Primary: 58J52, 47A10, 30F45, Secondary:  14H81}

\section{Introduction}\label{intro}
Investigation of determinants of Laplacians as functions of metrics on compact Riemann surfaces is motivated by the needs of geometric analysis and quantum field theory. For smooth metrics the determinants have been comprehensively studied, see e.g.~\cite{OPS,Weisberger,Wentworth,Wentworth2,Polch, Khuri,Kim}. The Polyakov formula~\cite{Pol,PolF} and a similar formula for surfaces with boundary due to Alvarez~\cite{Alvarez} often appears as the key of an argument,  e.g.~\cite{OPS,Weisberger,Wentworth,Kim}. More recently significant progress was  achieved in the study of determinants of Laplacians for  flat (curvature zero) metrics with conical singularities, see e.g.~\cite{Khuri2, Au-Sal, Au-Sal 2,Spreafico,KokotKorot,ProcAMS,HKK1,HKK2}. Here, for instance, 
results in~\cite{Au-Sal, Au-Sal 2} can be interpreted as a generalization of  Polyakov-Alvarez and Polyakov formula to the case of flat metrics with conical singularities on a disk and on a sphere;
the main result in~\cite{ProcAMS} is a simple consequence of  an analog of Polyakov formula  and the results in~\cite{KokotKorot}. 
Some results were also obtained for determinants of Laplacians on constant positive curvature  surfaces with conical sigularities~\cite{Dowker,Spreafico3,IMRN,Bul}, on Riemann orbisurfaces that are the quotients of the hyperbolic plane by the action of cofinite finitely generated Fuchsian groups~\cite{TZ2,Teo}, and on surfaces  glued from copies of singularly deformed spheres~\cite{KalvinJGA}, but no  Polyakov-Alvarez type formulas for metrics other than smooth or singular flat were available until now.

In this paper we present and prove Polyakov-Alvarez type formulas relating the determinants of  Frie\-de\-richs selfadjoint extensions of Laplacians for a pair of  conformally equivalent metrics with conical singularities on a compact Riemann surface. In particular, we thus find how the determinants depend on the orders (angles) of conical singularities. 
The main results are presented in Theorem~\ref{main} after giving some preliminaries in Section~\ref{PMR} below. 
The proof is  mainly based on  explicit formulas for the spectral zeta function of  Friederichs Dirichlet Laplcian on a flat cone~\cite{Spreafico}, short time heat trace asymptotic expansions of  elliptic cone differential operators~\cite{BS2,GKM},    Burghelea-Friedlander-Kappeler (BFK)  type decomposition formulas~\cite{BFK,LeeBFKboundary}, and the classical Polyakov-Alvarez formula~\cite{Alvarez}. The dependence of the determinants of Laplacians on orders of conical singularities  is expressed in terms of Barnes double zeta functions, in Appendix we show that in the case of rational orders it simplifies to a linear combination of  gamma functions and the Riemann zeta function.  
For the smooth metrics our results reduce to the classical Polyakov and Polyakov-Alvarez formulas~\cite{Pol,PolF,Alvarez}.  

 We also illustrate our general results with examples: based on our Polyakov-Alvarez type formulas we independently deduce known and obtain new explicit formulas for determinants of Laplacians on 
singular  surfaces  with or without boundary.  We believe that these examples can be not only helpful for better understanding of main results, but also of independent interest.

 In the first example  we   study extremal properties of the determinant for the metrics of constant curvature  on a sphere with two conical singularities as the orders of conical singularities and the geodesic distance between them vary. We first recover and generalize some results in~\cite{Spreafico3,Klevtsov,KalvinJGA,IMRN}  and then show that on the metrics with fixed area $4\pi$ the determinant is unbounded from above and attains its local maximum on the metric of standard  round sphere. Note the similarity and difference between this result and the well-known result~\cite{OPS,Sarnak} on  extremals of determinants for smooth metrics: on a closed surface in any conformal class of smooth metrics  with fixed area the determinant attains its absolute maximum on a unique metric of constant curvature (in particular, on the metric of the standard round sphere if the surface is a sphere).   
 
  In the second example we deduce a formula for the determinant on polyhedra with spherical topology and then show that it is equivalent to the famous Aurell-Salomonson formula thus providing the latter one with  rigid mathematical proof.   In the third example  we find a formula for the determinant of Friederichs Dirichlet Laplacian on the constant curvature metric disks with a conical singularity at the center: for the non-flat metrics the result is new, the flat case was previously studied in~\cite{Spreafico} by a method based on separation of variables. 
  In our forth example we obtain a new formula for the determinant of Laplacians on hyperbolic spheres with conical singularities, recover a variational formula  from~\cite{KalvinJGA}, and compare our results to those in the recent study of quantum Hall states on singular surfaces~\cite{C-W}. Finally,  in the last example we present a general explicit formula for the determinant of Friederichs Laplacians on singular genus $g>1$  surfaces without boundary, this is a generalization of the results in~\cite{KokotKorot,ProcAMS}.  
   
      The reader whose interests are in the examples and applications may wish to begin by looking at Section~\ref{PMR} to get acquainted with the notation used and the main results obtained and then proceed to Section~\ref{appl}, where the examples are presented.

 Genus one examples will be considered elsewhere.  We only note that explicit formulas for determinants of Laplacians on genus one surfaces can be obtained by using results of this paper together with known explicit formulas for the determinant of Laplacian on the (smooth) flat tori~\cite{Polch,OPS} and the flat annulus~\cite{Weisberger}; in particular, one can expect to recover the variational formula in~\cite{Bul} and to find the corresponding undetermined constant.  Let us also mention  that metrics with cylindrical and cone/Euclidean ends can be included into consideration by pairing results of this paper with the BFK-type decomposition formulas in~\cite[Theorem 1]{HKK1} and~\cite[Theorem 1]{HKK2}, however we do not discuss this here.

The paper is organized as follows.  Subsection~\ref{PMR} contains preliminaries and the main results. 
In Subsection~\ref{Sec Cor} we formulate  two important corollaries: a formula for the value of spectral zeta function at zero and Polyakov-Alvarez type formulas for two conformally equivalent conical metrics. The first corollary is an immediate consequence of our main result and the Gauss-Bonnet theorem. The proof of other results is carried out  in Section~\ref{s2}. In Subsection~\ref{ss2.1} we discuss local regularity of metrics. In Subsection~\ref{SSnon-flat} we obtain an asymptotic estimate for the determinant of the Friederichs Dirichlet Laplacian on a shrinking metric disk. In Subsection~\ref{sBFK} we prove BFK-type decomposition formulas. These decomposition formulas allow to cut shrinking metric disks out of a  surface. In Subsection~\ref{Pr} we finilize the proof of Polyakov-Alvarez type formulas.  This completes  the first part of the paper.

 The second part of the paper entitled ''Examples and Applications`` occupies Section~\ref{appl}. 
 In Subsection~\ref{AF} we study the determinant on the metrics of constant curvature on a sphere with two conical singularities. In Subsection~\ref{FM} we consider polyhedral surfaces with spherical topology and  deduce  the famous Aurell-Salomonson formula.  In Subsection~\ref{CMD} we deduce a formula for the determinant of Friederichs Dirichlet Laplacian on constant curvature metric disks with a conical singularity at the center. 
 In Subsection~\ref{HS} we study the determinant of Laplacians on  singular hyperbolic spheres. Finally, in Subsection~\ref{KOKOT} we present a general formula for the determinant of Friederichs Laplacians on genus $g>1$  surfaces with conical singularities and without boundary.

 \subsection{Preliminaries and main results}\label{PMR}

  Let $M$ be a compact Riemann surface (perhaps with smooth boundary $\partial M$). Following~\cite{Troyanov Curv,Troyanov Polar Coordinates}, we say that $m$ is a (conformal Riemannian)  metric on $M$  if for any point $P\in M$ there exist a neighbourhood $U$ of $P$, a local   (holomorphic) parameter $x\in \Bbb C$ centred at $P$ (i.e. $x(P)=0$), and a real-valued function $\phi\in L^1(U)$ 
  such that $m=|x|^{2\beta}e^{2\phi}|dx|^2$ in $U$ with some $\beta>-1$, and $\partial_x\partial_{\bar x} \phi\in L^1(U)$. If $\beta=0$, then the point $P$ is regular. If $\beta\neq 0$, then $P$ is a conical singularity of order $\beta$ and total angle $2\pi(\beta+1)$.  A function $K:M\to \Bbb R$ defined by
 $$
K=|x|^{-2\beta}e^{-2\phi}(-4\partial_x\partial_{\bar x} \phi)
$$    
 is the (regularized) Gaussian curvature of $m$ in the neighbourhood $U$ ($K$ does not depend on the choice of $x$).  
 
 The information about all conical singularities of $m$ is encoded in a divisor:  a metric with conical singularities of order $\beta_1,\dots,\beta_n$ at (distinct) points $P^1,\dots P^n\in M$ is said to represent the divisor $\pmb \beta=\sum_j \beta_j P^j$, which is a formal sum. By definition, the set $\supp\pmb \beta:=\{P^1,\dots,P^n\}$  is  the support  and the number $|\pmb \beta|:=\sum_j \beta_j $ is the degree of the divisor  $\pmb\beta$.  We assume that the curvature $K$ is a smooth function on $M$, if $M$ has a boundary $\partial M$, then $\supp\pmb \beta\cap\partial M=\varnothing$   (i.e. there are no conical singularities on the boundary)   and the geodesic curvature $k$ is a well defined continuous function on $\partial M$.

 Let $m_\varphi$ be a  metric conformally equivalent to a smooth metric $m_0$ on $M$, i.e. $m_\varphi$ represents a divisor $\pmb \beta$   and  $m_\varphi=e^{2\varphi}m_0$ with some function   $\varphi \in C^\infty(M\setminus\supp\pmb \beta)$. In a  local  parameter centred at $P^j\in\supp\pmb\beta$ we have $m_\varphi=|x|^{2\beta_j}e^{2\phi_j}|dx|^2$ and $m_0=e^{2\psi_j}|dx|^2$, where  $\phi_j$ is a continuous 
  and $\psi_j$ is a smooth function in a small neighbourhood of $x=0$. In particular,  $\varphi(x)=\beta_j\log |x|+\phi_j(x)-\psi_j(x)$ and, if the (regularized) Gaussian curvature  of $m_\varphi$ is not constant in any small neighbourhood of $x=0$,   an  assumption on local  regularity of $\phi_j$ near $x=0$  should be made, for details we refer to  Subsection~\ref{ss2.1}.   In what follows it is important that the value   $\frac{\phi_j(0)}{\beta_j+1}-\psi_j(0)$ does not depend on the choice of local  parameter $x$, $x(P^j)=0$. 

   Let $\Delta^\varphi$ stand for the Friederichs extension of the Laplacian on $(M,m_\varphi)$  initially defined on the functions in $C_0^\infty(M\setminus\supp\pmb\beta)$.  The spectrum $\sigma(\Delta^\varphi)$ of $\Delta^\varphi$  consists of isolated eigenvalues $\lambda_k$ of finite multiplicity. If $\partial M=\varnothing$, then the first eigenvalue $\lambda_0=0$ of the nonnegative selfadjoint operator $\Delta^\varphi$ is of multiplicity $1$ (and the eigenspace consists of constant functions). If  $\partial M\neq\varnothing$, then  $\Delta^\varphi$  is the Friederichs Dirichlet Laplacian, it is a positive selfadjoint operator. 
From results in~\cite{BS2,Seeley,Troyanov Polar Coordinates} it follows that  the spectral zeta function  
$$ \zeta(s)=\sum_{\lambda_k\in\sigma(\Delta^\varphi),\lambda_k\neq 0}\lambda_k^{-s},\quad\Re s>1,$$
  extends by analyticity to a neighborhood of $s=0$. We define the zeta regularized determinant of $\Delta^\varphi$  by  
  $$\det\Delta^\varphi=e^{-\zeta'(0)};$$ if $\partial M=\varnothing$, then it is a modified determinant, i.e. with zero eigenvalue excluded. For a smooth metric $m_0$ on $M$ the determinant $\det\Delta^0$ can be  defined via the spectral zeta function of the selfadjoint Laplacian $\Delta^0$  on $(M,m_0)$ in exactly the same way, e.g.~\cite{OPS}.

The main result of the first part of this  paper is the following generalization of Polyakov and Polyakov-Alvarez formulas. 

\begin{theorem}[Polyakov-Alvarez type comparison formulas]\label{main}
Let $m_0$ be a smooth conformal  metric on a compact Riemann surface $M$. Denote the Gaussian curvature of $m_0$ by  $K_0$.  Let  $K_\varphi$ stand for  the (regularized) Gaussian  curvature of  a  metric  $m_\varphi=e^{2\varphi}m_0$ representing a divisor $\pmb\beta=\sum_{j=1}^n \beta_j P^j$. By $\phi_j(x)$ and $\psi_j(x)$ we denote the functions in the representations $m_\varphi=|x|^{2\beta_j}e^{2\phi_j}|dx|^2$ and $m_0=e^{2\psi_j}|dx|^2$ in a local holomorphic parameter $x$ centred at $P^j\in\supp\pmb\beta$. 
Let also $A_\varphi$ (resp. $A_0$) stand for  the total area of $M$  in the metric $m_\varphi$ (resp. $m_0$). 
  
1. If $\partial M=\varnothing$, then for the modified zeta regularized determinants of the Friederichs Laplacian $\Delta^\varphi$ on $(M,m_\varphi)$ and the selfadjoint Laplacian $\Delta^0$ on $(M,m_0)$ we have 
 \begin{equation}\label{Pol}
 \begin{aligned}
\log\frac{(\det \Delta^\varphi)/A_\varphi}{(\det \Delta^0)/A_0}=-\frac{1}{12\pi}\left(\int_{M} K_{\varphi}\varphi\,dA_{\varphi}\right.&\left.+\int_{M} K_0\varphi\,dA_0\right) \\ +\frac  1 6 {\sum_{j=1}^n\beta_j}\left(\frac{\phi_j(0)}{\beta_j+1}-\psi_j(0)\right)&-\sum_{j=1}^nC(\beta_j).
\end{aligned}
\end{equation}
Here \begin{equation}\label{Cbeta}
C(\beta)=2\zeta'_B(0;\beta+1,1,1)-2\zeta'_R(-1) -\frac {\beta^2} {6(\beta+1)}\log 2-\frac {\beta} {12}+\frac 1 2 \log(\beta+1), 
\end{equation}
where  $\zeta_B$  is the Barnes double zeta function
\begin{equation}\label{ZB}
\zeta_B(s;a,b,x)=\sum_{m,n=0}^\infty(am+bn+x)^{-s},
\end{equation}
the  prime stands for the derivative with respect to $s$, and $\zeta_R(s)$ is the Riemann zeta function. 

2. If $\partial M\neq \varnothing$ and $\supp\pmb \beta\cap\partial M=\varnothing$, then for the zeta regularized determinants of the  Friederichs Dirichlet Laplacian $\Delta^\varphi$ on $(M,m_\varphi)$ and the selfadjoint Dirichlet Laplcian  $\Delta^0$ on $(M,m_0)$ we have
\begin{equation}\label{Pol-Alv}
\begin{aligned}
 \log  \frac{\det \Delta^\varphi}{\det\Delta^0}=-&\frac{1}{12\pi}\left(\int_M K_\varphi \varphi\,dA_\varphi +\int_M K_0\varphi\,dA_0 +\int_{\partial M} \varphi\partial_{\vec n}\varphi\,ds_0\right)
\\
& -\frac 1{6\pi}\int_{\partial M}k_0\varphi\, ds_0-\frac 1 {4\pi}\int_{\partial M} \partial_{\vec n}\varphi\,ds_0 \\&+\frac  1 6 {\sum_{j=1}^n\beta_j}\left(\frac{\phi_j(0)}{\beta_j+1}-\psi_j(0)\right)-\sum_{j=1}^nC(\beta_j),
\end{aligned}
\end{equation}
where $k_0$ is the geodesic curvature of the boundary $\partial M$ of $M$, $s_0$ is the arc length, and $\vec n$ is the outward unit normal to the boundary $\partial M$ (all are with respect to the metric $m_0$); the function $C(
\beta)$ is the same as in~\eqref{Cbeta}.

\end{theorem}

The non-integral terms in the right hand side of~\eqref{Pol} and~\eqref{Pol-Alv} are responsible for the inputs from the conical singularities and  do not depend on the choice of local parameters. 
If the function $\varphi$ in the equality  $m_\varphi=e^{2\varphi}m_0$ is smooth (or, equivalently, the metric $m_\varphi$ is smooth,  $\supp\pmb\beta=\varnothing$, and $n=0$), then the non-integral terms (the last lines in~\eqref{Pol} and~\eqref{Pol-Alv})  disappear. As a result the formulas~\eqref{Pol} and~\eqref{Pol-Alv} become  the well-known Polyakov and Polyakov-Alvarez formulas written in a slightly different ``regularized'' form, cf. e.g.~\cite{Pol,PolF,Alvarez,OPS,Weisberger}.  This regularization keeps the integrals in~\eqref{Pol} and~\eqref{Pol-Alv} finite even if  $m_\varphi$ has conical singularities.

For the conical singularities of rational orders $\beta$  the values of $C(\beta)$ in~\eqref{Pol} and~\eqref{Pol-Alv}  can be expressed in terms of $\zeta_R'(-1)$ and gamma functions. Namely, the following equality is valid for the derivative  of the Barnes double zeta function in~\eqref{Cbeta} :
\begin{equation}\label{wed sep4}
 \begin{aligned}
\zeta'_B(0;p/q,1,1)=&\frac 1 {pq}\zeta_R'(-1)-\frac{1}{12pq}\log(q)
+\left(\frac 1 4 {+}S(q,p)\right)\log{\frac{q}{p}}
\\
+\sum_{k=1}^{p-1}\left(\frac 1 2 -\frac k p \right)&\log \Gamma\left(  \left(\!\!\!\left(\frac{kq}{p}\right)\!\!\!\right)+\frac 1 2 \right)+\sum_{j=1}^{q-1}\left(\frac 1 2 -\frac j q \right)\log \Gamma\left(\left(\!\!\!\left(\frac{jp}{q}\right)\!\!\!\right)+\frac 1 2\right),
\end{aligned}
\end{equation}
 where $p$ and $q$ are coprime natural numbers, $S(q,p)=\sum_{j=1}^{p}\left(\!\!\left(\frac{j}{p}\right)\!\!\right) \left(\!\!\left(\frac{jq}{p}\right)\!\!\right)$ is the Dedekind sum, and the symbol $(\!(\cdot)\!)$ is defined so that   $(\!(x)\!)=x-\lfloor x\rfloor-1/2$ for $x$ not an integer and $(\!(x)\!)=0$ for $x$ an integer (here $\lfloor x\rfloor$ is the floor of $x$, i.e. the largest integer not exceeding $x$). In particular, $\zeta'_B(0;1,1,1)=\zeta'_R(-1)$ and~\eqref{Cbeta} gives $C(0)=0$; recall that $\beta=0$ for a regular point $P\in M$. Similarly, for a singularity of order $\beta=1$ (i.e. of angle $4\pi$) we have  $$C(1)=-\zeta_R'(-1)-\frac 1 {12}\log 2 -\frac 1 {12},$$ for a singularity of order $\beta=-1/2$ (i.e. of angle $\pi$) we have $$C\left(-1/ 2\right)=-\zeta_R'(-1)-\frac 1 {6}\log 2 +\frac 1 {24},$$ and etc. A proof of~\eqref{wed sep4} and some particular values of $\zeta'_B(0;p/q,1,1)$ can be found in  Appendix~\ref{BarnesPV}. We also note that  in the case $q=1$ the formula~\eqref{wed sep4} simplifies to~\eqref{BLF1}, and in the case $p=1$ it simplifies to~\eqref{BLF2}. 
Available asymptotics of the Barnes double zeta function  (e.g.~\cite{Matsumoto,Spreafico2}) imply that  $C(\beta)\to+\infty$ as $\beta\to- 1^+$ and $C(\beta)\to-\infty$ as $\beta\to+\infty$, see Fig.~\ref{FF} for a graph of $C(\beta)$.

\begin{figure} \centering\includegraphics[scale=.25]{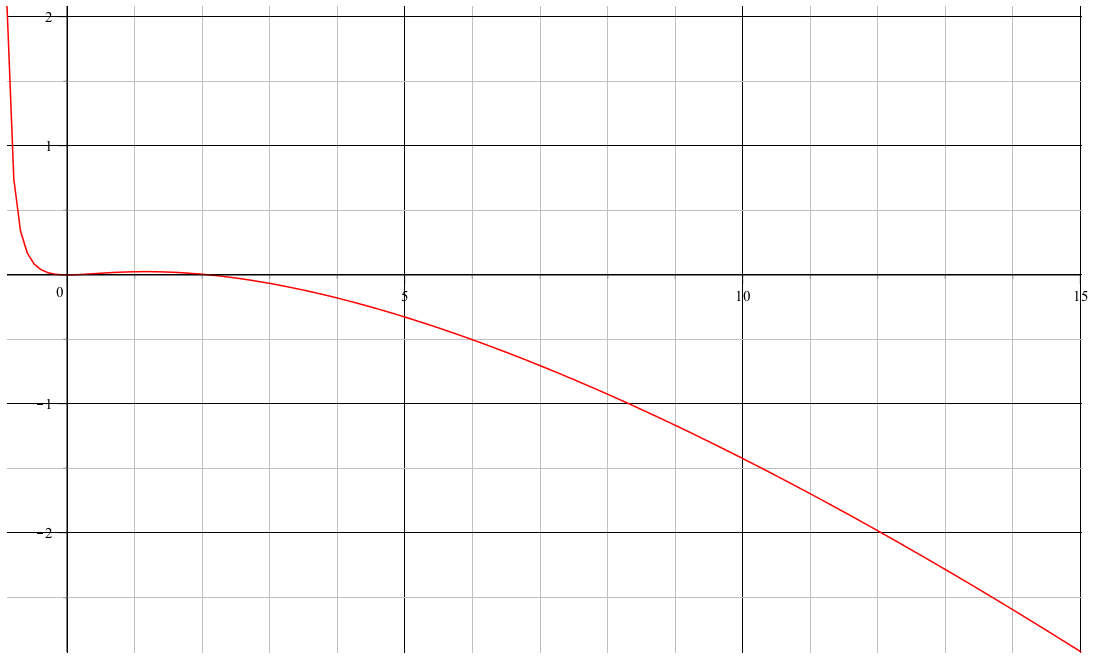}
 \caption{Graph of $\beta\mapsto C(\beta)$,  $\beta>-1$.}
 \label{FF}
\end{figure}

\subsection{A formula for $\zeta(0)$ and comparison formulas for two singular metrics}
\label{Sec Cor}

In this subsection we discuss two  corollaries  of Theorem~\ref{main}.  First we  find the value of the spectral zeta function of $\Delta^\varphi$ at zero. Then we present a generalization of Theorem~\ref{main} to the case of two  metrics with conical singularities.

Let $\zeta(s)=\sum_{k}\lambda_k^{-s}$ stand for the spectral zeta function  of the Friederichs Laplacian $\Delta^\varphi$. Then 
 $\zeta_r(s)=\sum_{k} (r^{-2}\lambda_k)^{-s}
 $
is the zeta function of the operator $r^{-2}\Delta^\varphi$ corresponding to the metric $r^2m_\varphi$. On the one hand, differentiating $\zeta_r(s)$ with respect to $s$ and evaluating the result at $s=0$ we arrive at the standard rescaling property
\begin{equation}\label{ReP}
\zeta'_r(0)=(2\log r)\zeta(0)+\zeta'(0),\quad r>0. 
\end{equation}
On the other hand, the Polyakov formula~\eqref{Pol} gives
$$
\zeta'_r(0)-\log(r^2 A_\varphi)-\zeta'(0)+\log A_\varphi=-\frac {\log r} {12\pi}\left(\int_M K_\varphi\,dA_\varphi+\int_M K_0\,dA_0\right)+\frac 1 6 \sum_{j=1}^n\frac {\beta_j\log r}{\beta_j+1},
$$ 
where $\int_M K_\varphi\,dA_\varphi=2\pi\chi(M,\pmb\beta)$ and $ \int_M K_0\,dA_0=2\pi\chi(M)$  by  the Gauss-Bonnet theorem~\cite{Troyanov Curv}; here $\chi(M,\pmb{\beta})=\chi(M) +|\pmb\beta|$ is the Euler characteristic of  $M$ with topological Euler characteristic $\chi(M)$  and divisor $\pmb\beta$.
This  implies
\begin{equation}\label{zeta_at_zero}
\zeta(0) =\frac {\chi(M,\pmb{\beta})} 6-\frac 1 {12}\sum_{j=1}^n\left(\beta_j+1-\frac{1}{\beta_j+1}\right)
-\dim\ker\Delta^\varphi
\end{equation}
in the case $\partial M=\varnothing$. Similarly, the Polyakov-Alvarez formula~\eqref{Pol-Alv}  implies~\eqref{zeta_at_zero} in the case $\partial M\neq \varnothing$. We formulate this result as a corollary of Theorem~\ref{main}. 

\begin{corollary}\label{z0} Let $\chi(M,\pmb{\beta})=\chi(M) +|\pmb\beta|$ stand for the Euler characteristic of the Riemann surface $M$ with topological Euler characteristic $\chi(M)$  and divisor $\pmb\beta$  of degree $|\pmb\beta|=\sum_{j=1}^n\beta_j$. Let $m_\varphi$ be a conformal metric on $M$ representing the divisor $\pmb \beta$.
Then the spectral zeta function $\zeta(s)$ of the Friederichs Laplacian $\Delta^\varphi$ on $(M,m_\varphi)$
satisfies~\eqref{zeta_at_zero}, where   $\dim\ker\Delta^\varphi=1$ in the case $\partial M=\varnothing$ and $\dim\ker\Delta^\varphi=0$ in the case $\partial M\neq\varnothing$.
\end{corollary}
We also note that the formula~\eqref{zeta_at_zero} for $\zeta(0)$  allows to find the constant term in the asymptotic expansion of the heat trace $\Tr e^{-t \Delta^\varphi}$ as $t\to 0+$ (see Remark~\ref{a_0} in Sec.~\ref{sBFK}).

The next  corollary presents  comparison formulas for the determinants of Laplacians in two singular metrics.
\begin{corollary}\label{2conical}
Let  $m_0$  and $m_\varphi=e^{2\varphi}m_0$ be two conformal   metrics on $M$ representing divisors $\pmb \alpha$ and $\pmb\beta$ respectively. 
Let $\{P_1,\dots P_n\}$ be the set of all distinct points in the union $\supp\pmb\alpha\cup\supp\pmb\beta $. 
In a local holomorphic parameter $x$ centred at $P^j$ we have  $m_0=|x|^{2\alpha_j}e^{2\psi_j}|dx|^2$  and $m_\varphi=|x|^{2\beta_j}e^{2\phi_j}|dx|^2$, where $\alpha_j=0$ if $P^j\notin\supp\pmb \alpha$ and   $\beta_j=0$ if  $P^j\notin\supp\pmb\beta$.  

1. If $\partial M=\varnothing$, then the determinants of the Friederichs Laplacians $ \Delta^\varphi$ and $\Delta^0$  satisfy 
\begin{equation}\label{18:09}
\begin{aligned}
\log\frac{(\det \Delta^\varphi)/A_\varphi}{(\det \Delta^0)/A_0}=-\frac{1}{12\pi}\left(\int_{M} K_{\varphi}\varphi\,dA_\varphi +\int_{M} K_0 \varphi \,dA_0\right)
\\
 + \frac  1 6 \sum_{j=1}^{n}\left\{\beta_j\left(\frac{\phi_j(0)}{\beta_j+1}- \psi_j(0)\right)-\alpha_j\left(\frac{\psi_j(0)}{\alpha_j+1}-\phi_j(0)\right) \right\}\\-\sum_{j=1}^{n}\Bigl(C(\beta_j)-C(\alpha_j)\Bigr).
\end{aligned}
\end{equation}

2. If $\partial M\neq\varnothing$, $\supp\pmb\alpha\cap\partial M=\varnothing$, and $\supp\pmb\beta\cap\partial M=\varnothing$, then the determinants of the Friederichs Dirichlet Laplacians  $ \Delta^\varphi$ and $\Delta^0$   satisfy
\begin{equation*}
\begin{aligned}
\log\frac{\det \Delta^\varphi}{\det \Delta^0}=-\frac{1}{12\pi}\left(\int_{M} K_{\varphi}\varphi\,dA_\varphi +\int_{M} K_0 \varphi\,dA_0\right.
\left.+\int_{\partial M}  \varphi\partial_{\vec n} \varphi\,ds_0\right)
\\ -\frac 1{6\pi}\int_{\partial M}k_0\varphi\, ds_0-\frac 1 {4\pi}\int_{\partial M} \partial_{\vec n}\varphi\,ds_0 
 \\  + \frac  1 6 \sum_{j=1}^{n}\left\{\beta_j\left(\frac{\phi_j(0)}{\beta_j+1}- \psi_j(0)\right)-\alpha_j\left(\frac{\psi_j(0)}{\alpha_j+1}-\phi_j(0)\right) \right\}\\-\sum_{j=1}^{n}\Bigl(C(\beta_j)-C(\alpha_j)\Bigr),
 \end{aligned}
\end{equation*}
where $k_0$ is the geodesic curvature of the boundary $\partial M$, $s_0$ is the arc length, and $\vec n$ is the outward unit normal to the boundary $\partial M$ (with respect to the metric $m_0$). \end{corollary}

  In the particular case of two  conformally equivalent  flat  metrics  on a surface without boundary  the formula~\eqref{18:09} with $K_\varphi=K_0=0$  is essentially an invariant reformulation  (i.e. independent of the choice of local holomorphic parameters)   of   the result in~\cite[Prop.1]{ProcAMS} with explicitly specified in~\eqref{Cbeta} function $C(\cdot)$ and the unnecessary additional assumption $\supp\pmb\alpha\cap\supp\pmb\beta=\varnothing$ removed.  
The proof of Corollary~\ref{2conical} is postponed to Section~\ref{Pr}.

\section{Proof of Polyakov-Alvarez type comparison formulas}\label{s2}

\subsection{Local regularity of  metrics near conical singularities}\label{ss2.1}
Recall that  for any flat in a neighbourhood of its conical singularity $P^j$  metric  $m_\varphi$   there exists a local holomorphic parameter $x$ such that $m_\varphi=|x|^{2\beta_j}|dx|^2$, see e.g.~\cite[Lemma 3.4]{Troyanov Polar Coordinates}. An important widely used property of the  flat  metric $|x|^{2\beta_j}|dx|^2$  is its homogeneity of degree $2$ in $\epsilon$ with respect to the dilations $x\mapsto\epsilon^{\frac{1}{\beta_j+1}} x$ (or, equivalently, the corresponding property of Laplacians, see e.g.~\cite[Sec. 5]{Cheeger},~\cite[Sec. 3]{BS2}, $\kappa$-homogeneity in~\cite{GKM}). 
     A non-flat in a neighborhood of its conical singularity $P^j$  metric  $m_\varphi=|x|^{2\beta_j}e^{2\phi_j}|dx|^2$  is no longer homogeneous with respect to the dilations, but it is natural to consider it as a local perturbation of  the flat  metric $|x|^{2\beta_j}|dx|^2$.  Then an  assumption on local regularity of the metric potential $\phi_j$  has to be made in order to specify the class of admissible perturbations. Before anything else, it is important to include into consideration all constant curvature metrics.  
  As is known, all  metrics of constant curvature  (even locally: near the conical singularities)  are (local) dilation analytic perturbations of flat  metrics, see e.g.~\cite{Br87,BinXuSp,BinXuHyp,MondelloPanov,Troyanov Polar Coordinates} and Remark~\ref{CC}  below. For all other metrics with conical singularities  we assume that they are admissible (local dilation analytic perturbations of the flat metrics) in the sense of the following 
 
  \begin{definition}\label{D A}
   A  metric $m_\varphi$ representing a divisor $\pmb \beta$ is admissible, if  for any $P^j\in\supp\pmb \beta$ there exists a centred at $P^j$ local holomorphic parameter $x$ such that the potential $\phi_j$ in the local representation $m_\varphi=|x|^{2\beta_j}e^{2\phi_j}|dx|^2$ viewed as  the function 
   $
   \epsilon\mapsto \phi_j(\epsilon^{\frac{1}{\beta_j+1}}x)
   $ 
   extends by analyticity from some  interval $[0,a]\subset\Bbb R$  to a (complex) neighbouhood of zero for any $x$ with $|x|<c$, where  $c>0$ is sufficiently small. 
  \end{definition}

\subsection{Dirichlet Laplacian on a shrinking metric  disk}\label{SSnon-flat}

Let $m_\varphi$ be a metric representing a divisor $\pmb \beta$.  Pick a point $P\in\supp\pmb \beta$ and  a local holomorphic parameter $x$ centred at $P$. 
 Consider the Friederichs Dirichlet Laplacian $\Delta^\varphi_{D_\epsilon}$ on the disk 
 $$
 D_\epsilon=\{x\in\Bbb C: |x|\leq \epsilon\}
 $$
  endowed with the singular metric $m_\varphi=|x|^{2\beta} e^{2\phi}|dx|^2$. More precisely,  $\Delta^\varphi_{D_\epsilon}$ is the Friederichs selfadjoint extension of the operator $-|x|^{-2\beta}e^{-2\phi}4\partial_x\partial_{\bar x}$ on $C_0^\infty(0<|x|\leq\epsilon)$ in the $L^2$-space   with the norm
$$
\|f\|=\left(\int_{|x|\leq\epsilon}|f(x,\bar x)|^2|x|^{2\beta} e^{2\phi}\frac {dx\wedge d\bar x}{-2i}\right)^{1/2}.
$$
In this section we obtain an asymptotic estimate for $\det \Delta^\varphi_{D_\epsilon}$ as $\epsilon\to 0+$ (Lemma~\ref{KeyLemma} below). For the Dirichlet Laplacian  $\Delta^\psi_{D_\epsilon}$ on the smooth metric disk $(D_\epsilon, e^{2\psi}|dx|^2)$  the  corresponding result can be easily obtained from the usual Polyakov-Alvarez formula~\cite{Alvarez,OPS,Weisberger} and the explicit formula~\cite[formula (28)]{Weisberger}  for the determinant of Dirichlet Lalacian on the flat metric disk  $(D_\epsilon, |dx|^2)$ (see Lemma~\ref{SimpleKey} at the end of this section).  
\begin{remark}\label{CC} As is known, for any constant  curvature $K_\varphi$ metric $m_\varphi$ on $M$  and any point $P\in M$ there exists a local holomorphic parameter $x$ such that $x(P)=0$ and in a neighborhood of $P$ one has
$$
m_\varphi=|x|^{2\beta}e^{2\phi}|dx|^2 \quad\text{with}\quad   \phi(x)=\log(2\beta+2)-\log (1+K_\varphi |x|^{2\beta+2});
$$
see e.g.~\cite[Lemma 3.4]{Troyanov Polar Coordinates} for the case $K_\varphi=0$,~\cite[Prop. 4]{Br87} and \cite{BinXuSp} for the case $K_\varphi>0$, and~\cite{BinXuHyp,TZ1} for the case $K_\varphi<0$; see also~\cite{MondelloPanov}.
Thus  $m_\varphi$ is an admissible (in the sense of Definition~\ref{D A})  dilation analytic  perturbation of the flat  metric $|x|^{2\beta}|dx|^2$. Moreover, the metric disks $(D_\epsilon,m_\varphi)$  and  $(D_1,\varepsilon^2 |x|^{2\beta}e^{2\phi(\varepsilon,x)}|dx|^2)$ with $\varepsilon=\epsilon^{\beta+1}$  and
$\phi(\varepsilon,x )=
\phi(\epsilon x)$  are isometric. Clearly, the metric on the unit disk $D_1$  approaches  the flat  metric $4\varepsilon^2(\beta+1)^2 |x|^{2\beta}|dx|^2$  as $\epsilon\to 0$. Therefore it is natural to expect that  after an appropriate  rescaling (cf.~\eqref{ReP}) the determinant $\det \Delta^\varphi_{D_\epsilon}$  approaches the one of the Friederichs Dirichlet Laplacian on the unit disk $D_1$ endowed with flat  metric $4|x|^{2\beta}|dx|^2$.  We further develop these ideas in Lemma~\ref{KeyLemma} below.
\end{remark}

  Let $\zeta_<(s,\beta)$   stand for the spectral zeta function of the Friederichs Dirichlet Laplacian on the unit disk $D_1$ with flat  metric $4|x|^{2\beta}|dx|^2$. 
  As is known~\cite{Spreafico,Spreafico2}, the  function $s\mapsto \zeta_<(s,\beta)$ admits an analytic continuation to $s=0$ and  
\begin{equation}\label{zeta at zero+}
 \zeta_<(0,\beta)= \frac 1 {12}\left(\beta+1+\frac 1 {\beta+1}\right),
\end{equation}
\begin{equation}\label{Spr+}
\begin{aligned}
\zeta'_<(0,\beta)=2\zeta'_B(0;\beta& +1,1,1)
+\frac {5} {12}(\beta+1)+\frac 1 2 \log (\beta+1)  +\frac 1 2 \log 2\pi,
\end{aligned}
\end{equation}
where the prime stands for the derivative with respect to $s$ and  $\zeta_B$ is the Barnes double zeta function~\eqref{ZB};  see also~\cite{Klevtsov}.

\begin{lemma}\label{KeyLemma}  For the spectral determinant of the Friederichs Dirichlet Laplacian $\Delta^\varphi_{D_\epsilon}$  on the  disk $D_\epsilon=\{x\in\Bbb C: |x|\leq \epsilon\}$  endowed with  metric $m_\varphi=|x|^{2\beta} e^{2\phi}|dx|^2$, where $x$ is a local holomorphic parameter from Definition~\ref{D A},  we have
\begin{equation}\label{E1s++}\begin{aligned}
\log\det\Delta^\varphi_{D_\epsilon}= & 2\Bigl(\log(2\epsilon^{-\beta-1})-\phi(0)\Bigr)\zeta_<(0,\beta)
\\
&  -\zeta_<'(0,\beta)+ O\bigl(-\epsilon^{\beta+1} \log\epsilon\bigr)\quad  \text{ as } \epsilon\to 0+,
\end{aligned}
\end{equation}
where $\zeta_<(0,\beta)$  and $\zeta'_<(0,\beta)$ are the same as in~\eqref{zeta at zero+} and~\eqref{Spr+}. 
\end{lemma}

\begin{proof} Denote $\varepsilon=\epsilon^{\beta+1}$ and let $\phi(\varepsilon,x )=
\phi(\epsilon x)$. Since the metric disks $(D_\epsilon,m_\varphi)$  and  $(D_1,\varepsilon^2 |x|^{2\beta}e^{2\phi(\varepsilon,x)}|dx|^2)$ are isometric, we can replace $\Delta^\varphi_{D_\epsilon}$ by the Friederichs extension $\Delta^\varphi_\varepsilon$ of the Dirichlet Laplacian $-\varepsilon^{-2}|x|^{-2\beta}e^{-2\phi(\varepsilon,x)}4\partial_x\partial_{\bar x}$  on the unit disk $D_1$. 

Introduce an appropriate rescaling by setting $\hat\Delta^\varphi_\varepsilon=\frac 1 4 \varepsilon^2e^{2\phi(0)}\Delta^\varphi_\varepsilon$.  
In a small neighbourhood of zero $\varepsilon\mapsto \hat\Delta^\varphi_\varepsilon$ is a type A~\cite{Kato} analytic family of operators  in the space  $L^2(|x|\leq 1, |x|^{2\beta}|dx|^2)$. In particular,  the selfadjoint operator  $\hat\Delta^\varphi_0$ corresponds to the flat  metric $4|x|^{2\beta}|dx|^2$ and $\zeta_<(s)$ is its spectral zeta function (for brevity of notations we do not list $\beta$ as an  argument of the zeta functions  in this proof). It is known that the spectrum $\sigma(\hat\Delta^\varphi_0)$ of $\hat\Delta^\varphi_0$ consists of isolated eigenvalues $\lambda_k$, 
$$
0<\lambda_0\leq\lambda_1\leq\lambda_2\leq\cdots\leq \lambda_k\to+\infty,
$$
 and
 $(\hat\Delta^\varphi_0)^{-2}$ is a trace class operator.  Since 
$$
\|(\hat\Delta^\varphi_0-\lambda)^{-1}\|= 1/\operatorname{dist}\{\lambda,\sigma(\hat\Delta^\varphi_0)\},
$$
 the first resolvent identity together with inequality $\|AB\|_1\leq \|A\|\|B\|_1$ implies 
$$
\|(\hat\Delta^\varphi_0-\lambda)^{-2}\|_1\leq \left( 1+|\lambda| \|(\hat\Delta^\varphi_0-\lambda)^{-1}\|\right)^2\|(\hat\Delta^\varphi_0)^{-2}\|_1\leq C
$$ 
uniformly in $\lambda\in\mathcal C$, where $\mathcal C$ is a set such that  $\operatorname{dist}\{\lambda,\sigma(\hat\Delta^\varphi_0)\}\geq c|\lambda|$ for any $\lambda\in\mathcal C$ and some $c>0$.  Denote 
$$
T(\varepsilon,\lambda)=(\hat\Delta^\varphi_\varepsilon-\hat\Delta^\varphi_0)(\hat\Delta^\varphi_0-\lambda)^{-1}
$$
and observe that  $\|T(\varepsilon,\lambda)\|\to 0$ uniformly in $\lambda\in\mathcal C$ and $\varepsilon$ as $|\varepsilon|\to 0$.  We have
\begin{equation}\label{ee1+}
 \|(\hat\Delta^\varphi_\varepsilon-\lambda)^{-2}\|_1=\|\bigl(\operatorname{Id}+T^*(\bar{\varepsilon},\bar\lambda)\bigr)^{-1} (\hat\Delta^\varphi_0-\lambda)^{-2}\bigl(\operatorname{Id}+T({\varepsilon}, \lambda)\bigr)^{-1}\|_1\leq C
\end{equation}
for all $\lambda\in\mathcal C$ and $|\varepsilon|<\delta\ll1$.  Introduce the spectral zeta function 
\begin{equation*}
\zeta_\prec(s,\varepsilon)=\frac 1 {2\pi i (s-1)}\int_{\mathcal C} \lambda^{1-s}\Tr (\hat\Delta^\varphi_\varepsilon-\lambda)^{-2}\, d\lambda,
\end{equation*}
where $\mathcal C$ is a contour running clockwise at a sufficiently close distance around the cut $(-\infty,0]$ and  $\lambda^z=|\lambda|^z e^{iz\arg z}$ with $|\arg \lambda|\leq \pi$, cf.~\cite{Shubin}. Then~\eqref{ee1+} implies that $(s,\varepsilon)\mapsto\zeta_\prec(s,\varepsilon)$ is an analytic function of $s$ for $\Re s>2$ and $\varepsilon$ for $| \varepsilon|<\delta\ll1$. One of the ways to see analyticity in $\varepsilon$ is to  make the substitution 
$$
\Tr (\hat\Delta^\varphi_\varepsilon-\lambda)^{-2}=\frac 1 {2\pi i }\sum_{k\geq0}\oint\frac  {\bigl((\hat\Delta^\varphi_\mu-\lambda)^{-2}\psi_k,\psi_k\bigr)}{\mu-\varepsilon}\, d\mu,
$$
where  $\psi_k$ is an orthonormal basis in  $L^2(|x|\leq 1, |x|^{2\beta}|dx|^2)$ and
$$
\sum |\bigr((\hat\Delta^\varphi_\mu-\lambda)^{-2}\psi_k,\psi_k\bigr)|\leq \|(\hat\Delta^\varphi_\mu-\lambda)^{-2}\|_1\leq C
$$
because of~\eqref{ee1+}. After the substitution one can  change the order of integration and summation to obtain  the Cauchy's integral formula for  $\varepsilon\mapsto\zeta_\prec(s,\varepsilon)$. 

In the remaining part of this proof we show that $(s,\varepsilon)\mapsto \zeta_\prec(s,\varepsilon)$  continues analytically to $(0,0)$. Then thanks to  $\zeta_\prec(s,0)=\zeta_<(s)$  we can conclude that 
\begin{equation}\label{MonAug5}
\zeta_\prec(0,\varepsilon)=\zeta_<(0)+O(\varepsilon), \quad \zeta'_\prec(0,\varepsilon)=\zeta'_<(0)+O(\varepsilon),
\end{equation}
where the prime stands for the derivative with respect to $s$.  The standard rescaling argument guarantees that multiplication of a metric by $R^2$ adds $\zeta(0)\log R^2$ to the corresponding value of $\zeta'(0)$; see Sec.~\ref{Sec Cor}. Since $\hat\Delta^\varphi_\varepsilon=\frac 1 4 \varepsilon^2e^{2\phi(0)}\Delta^\varphi_\varepsilon$, the  rescaling argument and~\eqref{MonAug5} lead to
\begin{equation*}
\log\det\Delta^\varphi_\varepsilon=2\Bigl(\log(2\varepsilon^{-1})-\phi(0)\Bigr)\zeta_<(0)  -\zeta_<'(0)+ O\bigl(-\varepsilon \log\varepsilon\bigr) \quad \text{ as } \varepsilon\to0\!+\!.
\end{equation*}
Taking into account the equality   $\varepsilon=\epsilon^{\beta+1}$ we arrive at~\eqref{E1s++}. 

It suffices to show that $s\mapsto\zeta_\prec(s,\varepsilon)$ continues  analytically  from  $\Re s>2$ to $s=0$ for each $\varepsilon$, $|\varepsilon|<\delta\ll1$. We will rely on the representation
\begin{equation}\label{Wed15:19}
\zeta_\prec(s,\varepsilon)=\frac 1 {\Gamma(s)}\int_0^\infty t^{s-1} \Tr \bigl(e^{-t\hat\Delta^\varphi_\varepsilon}\bigr)\, dt
\end{equation}
together with short  time heat trace asymptotics~\cite{BS2,GKM}. (For the large values of $t$  the estimate $|\Tr \bigl(e^{-t\hat\Delta^\varphi_\varepsilon}\bigr)|=O(e^{-ct})$  with some $c>0$ immediately follows from $$e^{-t\hat\Delta^\varphi_\varepsilon}=\frac i{2\pi t}\int_{\mathcal C} e^{-\lambda t}(\hat\Delta^\varphi_\varepsilon-\lambda)^{-2}\,d\lambda$$ with a suitable contour $\mathcal C$ in the right half-plane  $\Re \lambda >0$  and~\eqref{ee1+}.)

In the polar coordinates $(r,\theta)=\bigl((\beta+1)^{-1}|x|^{\beta+1},\arg x\bigr)$ the  operator $\hat\Delta^\varphi_\varepsilon$ takes the form
$$
\hat\Delta^\varphi_\varepsilon=-e^{2\bigl(\phi(0)-\phi(\varepsilon r, \theta)\bigr)}(2r)^{-2}\Bigl((r\partial_r)^2+(\beta+1)^{-2}\partial_\theta^2\Bigr),
$$
where the function  $(r,\theta)\mapsto\phi( \varepsilon r,\theta)=\phi(\varepsilon,x)$ is smooth up to $r=0$. Therefore $\hat\Delta^\varphi_\varepsilon$ falls into the class of elliptic cone differential operators with stationary domains studied in~\cite{GKM}; we recall that the domain of $\hat\Delta^\varphi_\varepsilon$ coincides with the domain of the Friederichs extension $\hat\Delta^\varphi_0$ and the domains of Friederichs extensions are always stationary. Let $\chi(r,\theta)=\chi(r)$ with a cutoff  function $\chi\in C_c^\infty\bigl([0,\frac 1 {\beta+1})\bigr)$  that equals $1$ in a neighborhood of $r=0$.  A direct application of the main result in~\cite{GKM}\footnote{Proof of Thm. 4.4 in~\cite{GKM} requires some corrections~\cite{GK}: it should be $J=N+n+1$ (instead of $J=N+1$) on both places where the choice is relevant, the statement $\alpha(y,\hat\lambda)=0$ for $k<n$ on page 6511 is incorrect, on the same page the estimate $t_{N,N+n}(y,\lambda)=O(|\lambda|^{-N/m-\ell}\log|\lambda|)$ (that follows from the last equality on page 6510) is needed in addition to~(4.11) and (4.12). I would like to thank Juan B. Gil and Thomas Krainer for responding promptly to my inquiries about the proof and for sending me the corrected version of the paper.}  implies that  the  short time asymptotics of the heat trace $\Tr \bigl(\chi e^{-t\hat\Delta^\varphi_\varepsilon}\bigr)$ has the form
\begin{equation}\label{exp14}
 c_{-1} t^{-1}+c_{-1/2}t^{-1/2}+c_0+c_{01}\log t+ \sum_{k=1}^\infty \sum_{\ell=0}^{m_k} c_{k\ell} t^{\frac j 2} \log^\ell t\quad\text{ as } t\to 0+,
\end{equation}
where the coefficients $c_k$ and $c_{k\ell}$ depend on $\epsilon$. 
There are no conical singularities on the support of $(1-\chi)$ and hence the short time asymptotic expansion  
$$
\Tr \bigl((1-\chi) e^{-t\hat\Delta^\varphi_\varepsilon}\bigr)\sim\sum_{j\geq -2} C_j(\varepsilon) t^{j/2}\quad\text{ as } t\to 0+
$$ 
can be obtained in the standard well-known way, e.g.~\cite{Gilkey,Seeley}. 
In total  we have
$$
\Tr( e^{-t \hat\Delta^\varphi_\varepsilon})= a_{-1}(\varepsilon)t^{-1}+a_{-1/2}(\varepsilon) t^{-\frac 1 2}+a_0^0(\varepsilon)+a_0^1(\varepsilon)\log t+O(t^{1/2}\log^{m_1} t)\quad\text{ as } t\to 0+
$$
with some coefficients $a_{-1}$, $a_{-1/2}$, $a_0^0$, and $a_0^1$.
The representation~\eqref{Wed15:19} gives
\begin{equation}\label{zetaKeyL}
\zeta_\prec(s,\varepsilon)=\frac 1 {\Gamma(s)}\left(\frac{a_{-1}(\varepsilon)}{s-1}+\frac{a_{-1/2}(\varepsilon)}{s-1/2}+\frac{a_0^0(\varepsilon)}{s}-\frac{a_0^1(\varepsilon)}{s^2}+R(\varepsilon,s)\right),
\end{equation}
where $R(\varepsilon,s)$ is analytic in $s$ for $\Re s>-1/2$; recall that $1/\Gamma(s)=s+\gamma s^2+O(s^3)$. 

Thus
$
(s,\varepsilon)\mapsto s\zeta_\prec(s,\varepsilon)
$
continues analytically from $\Re s>2$, $|\varepsilon|<\delta\ll1$ to a neighbourhood of $(0,0)$. Moreover,  $s\zeta_\prec(s,\varepsilon)\bigr|_{s=0}=-a_0^1(\varepsilon)$. But results in~\cite{BS2,Troyanov Polar Coordinates} guarantee that $a_0^1(\varepsilon)=0$  (first for all  $\varepsilon\geq 0$, and then, by analyticity,  for all $\varepsilon$ with $|\varepsilon|<\delta\ll1$). 

Indeed, if $\varepsilon\geq 0$, then  $4 |x|^{2\beta}e^{2\bigl(\phi(\varepsilon,x)-\phi(0)\bigr)}|dx|^2$ is a  metric with conical singularity on the disk $|x|\leq 1$. By~\cite[Theorem 4.1]{Troyanov Polar Coordinates} in a small neighborhood of $x=0$  there exist smooth local geodesic polar coordinates $(\rho,\theta)$  such that  
 $$4 |x|^{2\beta}e^{2\bigl(\phi(\varepsilon,x)-\phi(0)\bigr)}|dx|^2=d\rho^2+h^2(\rho,\theta)d\theta^2,\quad \theta\in [0,2\pi(\beta+1)),$$ 
$$
\lim_{\rho\to0}\frac{h(\rho,\theta)}{\rho}=1,\quad  h_\rho(0,\theta)=1,\quad  h_{\rho\rho}(0,\theta)=0,
$$
where $ h_\rho=\partial_\rho h$ and  $h_{\rho\rho}=\partial_\rho^2 h$. Let $\chi(\rho,\theta)=\chi(\rho)$ be a smooth cutoff function supported in a small neigborhood of $\rho=0$ and such that $\chi(\rho)=1$ for all $\rho$ sufficiently close to $0$. Then $\chi\hat\Delta^\varphi_\varepsilon$ can be  considered as the  operator 
$
\chi h^{-1/2}\left(-\partial_\rho^2+\rho^{-2}\mathcal A(\rho)\right)h^{1/2}
$
 in the space $L^2(h(\rho,\theta)\,d\rho \,d \theta)$, where 
$$
\rho\mapsto \mathcal A(\rho)=-\rho^2\left(\frac{ h_\rho^2}{4h^{2}}-\frac{ h_{\rho\rho}} {2 h}  +h^{1/2}\left(\frac{1}{h}\partial_\theta\right)^2h^{-1/2}\right)
$$
is a smooth family of operators on the circle $\Bbb R/2\pi(\beta+1)\Bbb Z$. As a consequence, by~\cite[Theorem~5.2 and Theorem~7.1]{BS2} we have 
\begin{equation}\label{HA}
\Tr \chi e^{-\hat\Delta^\varphi_\varepsilon t}\thicksim\sum_{j=0}^\infty A_j t^{\frac {j-3}{2}} +\sum_{j=0}^\infty B_j t^{-\frac{\alpha_j+4}{2}}+\sum_{j: \alpha_j\in\Bbb Z_- } C_j  t^{-\frac{\alpha_j+4}{2}} \log t\quad \text{    as } t\to0+
\end{equation}
with some coefficients $A_j$, $B_j$, and $C_j$,  and an infinite sequence of numbers  $\{\alpha_j\}$ with $\Re\alpha_j\to-\infty$. The coefficient $C_j$ before $t^0\log t$ is given by
$\frac1 4\Res \zeta(-1)$, where $\zeta$ is the spectral zeta function of $(\mathcal A(0)+1/4)^{1/2}$; see \cite[formula (7.24)]{BS2}. Since $\mathcal A(0)=-\partial^2_\theta-1/4$,  we obtain
$$\zeta(s)=2\sum_{j\geq1} ( j/(2\beta+2))^{-s}=2 (2\beta+2)^{s}\zeta_R(s). $$
Thus $\Res \zeta(-1)=0$ and  the coefficient $C_j$ before  $t^0\log t$  is  zero.  This together with  $\Tr \bigl((1-\chi) e^{-t\hat\Delta^\varphi_\varepsilon}\bigr)\sim\sum_{j\geq -2} c_j t^{j/2}$ implies that the coefficient $a_0^1(\varepsilon)$ in~\eqref{zetaKeyL} is zero. Hence $s\mapsto\zeta_\prec(s,\varepsilon)$  continues analytically from  $\Re s>2$ to $s=0$ for each $\varepsilon$, $|\varepsilon|<\delta\ll1$.  This completes the proof. \end{proof}

\begin{lemma}\label{SimpleKey} Let $\Delta^\psi_{D_\epsilon}$ be the selfadjoint Dirichlet Laplacian on the metric disk $(D_\epsilon,e^{2\psi}|dx|^2)$, where $\psi$ is smooth. Then  
\begin{equation}\label{E2}
\log \det \Delta^\psi_{D_\epsilon}=\frac 1 3\Bigl( \log(2 \epsilon^{-1})-\psi(0)\Bigr)-
\zeta_<'(0,0)+O(\epsilon) \quad\text{as } \epsilon\to 0+
\end{equation}
with $\zeta_<'(0,\beta)$  given in~\eqref{Spr+}.
\end{lemma}
\begin{proof}
For the selfadjoint Dirichlet Laplacian $\Delta^0_{D_\epsilon}=-4\partial_x\partial_{\bar x}$ in the disk $|x|\leq\epsilon$ we have 
\begin{equation}\label{E1}
\log \det \Delta^0_{D_\epsilon}=-\frac 1  3 \log\epsilon +\frac 1 3 \log 2-\zeta'_<(0,0);\end{equation}
see~\eqref{zeta at zero+} and~\eqref{Spr+} with $\beta=0$ or~\cite[formula (28)]{Weisberger}. 
Since $\psi$ is smooth, we can use the classical Polyakov-Alvarez formula~\cite{Alvarez,OPS,Weisberger}, which  gives
\begin{equation*}\label{eprst}
\log \frac {\det \Delta^\psi_{D_\epsilon}}{\det \Delta^0_{D_\epsilon}}= -\frac{1}{6\pi}\left (\frac1 2 \int_{|x|\leq\epsilon} |\nabla_0\psi|^2\,dA_0+\int_{|x|=\epsilon}k_0\psi \,d s_0  \,\right )-\frac 1 {4\pi}\int_{|x|=\epsilon}\partial_{\vec n}\psi\,d s_0.
\end{equation*}
Here $\nabla_0$ is the gradient, $k_0=1/\epsilon$ is the geodesic curvature of the circle $|x|=\epsilon$,  and $\vec n$ is the outward unit normal to the disk  $|x|\leq \epsilon$ (all with respect  to the metric $|dx|^2$).  Therefore
$$
\log \frac {\det \Delta^\psi_{D_\epsilon}}{\det \Delta^0_{D_\epsilon}}=-\frac{1}{6\pi}\left (O(\epsilon)+\int_0^{2\pi}\Bigl(\psi(0)+O(\epsilon)\Bigr) \,d \theta  \,\right )-O(\epsilon)=-\frac 1 3 \psi(0)+O(\epsilon).
$$
This together with~\eqref{E1} completes the proof. 
\end{proof}

 \subsection{BFK decomposition formulas}\label{sBFK}
 By  $D_\epsilon^j\subset M$ we denote the $\epsilon$-neighborhood of conical singularity $P^j\in\supp\pmb\beta$ of $m_\varphi$ such that $P\in D_\epsilon^j$ if and only if  $|x(P)|\leq \epsilon$, where  $x$  is a local holomorphic parameter  from Definition~\ref{D A}. For   sufficiently small $\epsilon>0$ the disks  $D_\epsilon^1,\dots,D_\epsilon^N$ are disjoint and do not touch the boundary $\partial M$ of $M$. Let  $M_\epsilon=M\setminus \left\{D_\epsilon^1\cup\dots\cup D_\epsilon^n\right\}$ and let $\partial M_\epsilon$ stand for the boundary of $M_\epsilon$.   
 
 In this section we prove  BFK-type decomposition formulas for $\det\Delta^\varphi$ along the boundary $\partial M_\epsilon\setminus\partial M$ (Proposition~\ref{BFKf} below).  This is an analog of the BFK decomposition formula in~\cite[Theorem $\rm B^*$]{BFK} if $\partial M=\varnothing$ and of the one in~\cite[Corollary 1.3]{LeeBFKboundary} if $\partial M\neq \varnothing$. As is known,  for the metrics that are flat near the conical singularities  the BFK decomposition formulas and their proofs remain valid provided that one considers the Friederichs extensions of the Laplacians and the decomposition is done along a smooth closed curve that does not contain any  singularity of the metric;  see e.g.~\cite{HKK1,HKK2,ProcAMS,KokotKorot} and~\cite{LMP} for a more general result.  In our case the decomposition formulas are still  valid but their proof requires some minor modifications due to appearance of logarithmic terms in the short time heat trace asymptotics.    

As before, let $\Delta^\varphi$ stand for the Friederichs extension of the Laplacian on $(M,m_\varphi)$  ($\Delta^\varphi$ is the Dirichlet Laplacian if $\partial M\neq\varnothing$).  Consider also the Friederichs extension $\Delta^\varphi_{\partial M_\epsilon}$ of the Laplacian on $(M,m_\varphi)$ with Dirichlet boundary condition on $\partial M_\epsilon$; more precisely,  $\Delta^\varphi_{\partial M_\epsilon}$ is the Friederichs extension in $L^2(M,m_\varphi)$ of the operator $\Delta^\varphi$ defined on the functions  $u\in C^\infty_0\left(M\setminus \supp\pmb \beta\right)$ satisfying $u|_{\partial M_\epsilon}=0$. 

\begin{lemma}\label{hTr} 
The heat traces  $\Tr( e^{-t \Delta^\varphi})$ and $\Tr( e^{-t \Delta^\varphi_{\partial M_\epsilon}})$  have short time asymptotic expansions of the form 
\begin{equation}\label{hta}
 a_{-1}t^{-1}+a_{-1/2} t^{-1/2}+a_0 +\sum_{k=1}^\infty\sum_{\ell=0}^{m_k}a_{k \ell}t^{\frac{k}{2}}\log^\ell t\quad\text{ as } t\to 0+,\\
\end{equation}
where $a_{k}$,  and $a_{k\ell}$ are some coefficients.
If, in addition, the metric $m_\varphi$ is flat in a neighborhood of $\supp\pmb\beta$, then there are no logarithmic terms  in the asymptotic expansions (i.e. $m_k=0$  for all $k=1,2,3,\dots$). 
\end{lemma}
\begin{remark}\label{a_0}  As is known,  the constant term $a_0$ in the short time asymptotic expansion~\eqref{hta}  of the heat trace $\Tr( e^{-t \Delta^\varphi})$ is related to the value of the spectral zeta function of $\Delta^\varphi$ at zero by  $a_0=\zeta(0)+\dim\ker\Delta^\varphi$. As a consequence of Corollary~\ref{z0},  we thus obtain 
$$
a_0=\frac {\chi(M,\pmb{\beta})} 6-\frac 1 {12}\sum_{j=1}^n\left(\beta_j+1-\frac{1}{\beta_j+1}\right).
$$
\end{remark}
\begin{proof}[Proof of Lemma~\ref{hTr}] 
Introduce the local polar coordinates 
$$(r,\theta)=\bigl((\beta_j+1)^{-1}|x|^{\beta_j+1},\arg x\bigr)$$ centred at a conical singularity $P^j\in\supp\pmb\beta$; here $x(P^j)=0$ and  $x$  is a local holomorphic parameter  from Definition~\ref{D A}. In these coordinates  the Laplacian $ \Delta^\varphi$ takes the form 
$$
 \Delta^\varphi=-e^{\phi_j(r, \theta)}r^{-2}\Bigl((r\partial_r)^2+(\beta_j+1)^{-2}\partial_\theta^2\Bigr)
 $$
and thus falls into the class of elliptic cone operators with stationary domains studied in~\cite{GKM}: 1) The potential  $(r,\theta)\mapsto\phi_j(r,\theta)=\phi_j(x)$ is smooth up to $r=0$; 2) The domain of $\Delta^\varphi$ is stationary because $\Delta^\varphi$ is the Friederichs selfadjoint extension in $L^2(M,m_\varphi)$.
  
  Let   $\chi_j(r,\theta)=\chi_j(r)$ with a cutoff function $\chi_j\in C_c^\infty\bigl([0,\frac {\epsilon^{\beta_j+1}} {\beta_j+1})\bigr)$  that equals $1$ in a neighborhood of $r=0$; we extend $\chi_j$ from $D_\epsilon^j$ to $M$ by zero. Then a direct application of~\cite[Theorem 1.1]{GKM} implies that the heat trace $\Tr \bigl(\chi_j e^{-t\Delta^\varphi}\bigr)$ has a short time asymptotic expansion of the form~\eqref{exp14}. Moreover, if the potential $\phi_j$ does not depend on $r$ for all sufficiently small values of $r$, then in the expansion~\eqref{exp14} we have $m_k=0$ for all values of $k$ (note that  for a flat near $P^j$ metric $m_\varphi$  we can always achieve $\phi_j(r, \theta)=0$ by taking a suitable local holomorphic parameter $x$, e.g.~\cite[Lemma 3.4]{Troyanov Polar Coordinates}). 
This together with the standard well know expansion $$\Tr \bigl((1-\sum_j\chi_j) e^{-t\Delta^\varphi}\bigr)\sim\sum_{k\geq -2} C_k t^{k/2}\quad\text{ as } t\to 0+$$ implies~\eqref{hta} with extra term $a_0^1\log t$, cf.~\cite{Cheeger}.  Relying on~\cite{BS2,Troyanov Polar Coordinates} and using the same argument as in the proof of Lemma~\ref{KeyLemma}, one can verify that $a_0^1=0$, we omit the details. This  proof can also be repeated verbatim with $\Delta^\varphi$ replaced by $\Delta^\varphi_{\partial M_\epsilon}$.
\end{proof}

For $\lambda>0$ the operator  $\Delta^\varphi+\lambda$ is positive and hence $e^{-t\lambda}\Tr \bigl(e^{-t\Delta^\varphi}\bigr)=O(e^{-t\lambda})$  as $t\to +\infty$. Based on this  and~Lemma~\ref{hTr} we conclude that  the spectral zeta function
\begin{equation}\label{MK13:17}
\zeta(s,\lambda)=\Tr(\Delta^\varphi+\lambda)^{-s}=\frac 1 {\Gamma(s)}\int_0^\infty t^{s-1} e^{-t\lambda}\Tr \bigl(e^{-t\Delta^\varphi}\bigr)\, dt
\end{equation}
is holomorphic in $s$ for $\Re s>1$ and admits an analytic continuation to $s=0$ given by the right hand side of~\eqref{MK13:17}.
Therefore we can set $\det (\Delta^\varphi+\lambda)=e^{-\partial_s\zeta(0,\lambda)}$.
Similarly we define 
$\det (\Delta^\varphi_{\partial M_\epsilon}+\lambda)$.

Now we are in position to introduce the Neumann jump operator on $\partial M_\epsilon\setminus\partial M$.  For $\lambda>0$  and any $f\in C^\infty(\partial M_\epsilon\setminus\partial M)$ there exists a unique solution to the Dirichlet problem 
\begin{equation}\label{BVP}
(\Delta^\varphi+\lambda)u(\lambda)=0 \text{ on } M\setminus\partial M_\epsilon, \quad u(\lambda)=f \text { on } \partial M_\epsilon\setminus\partial M,\quad u(\lambda)=0 \text{ on } \partial M, 
\end{equation}
such that 
 $$
 u(\lambda)=\hat f -(\Delta^\varphi_{\partial M_\epsilon}+\lambda)^{-1}(\Delta^\varphi+\lambda)\hat f,
 $$
 where $\hat f\in C_0^\infty (M\setminus\supp\pmb\beta)$ is an extension of $f$. Introduce the Neumann jump operator $R^\varphi_\epsilon(\lambda): C^\infty(\partial M_\epsilon\setminus\partial M)\to C^\infty(\partial M_\epsilon\setminus\partial M)$ that acts by the formula
 $$
 R^\varphi_\epsilon(\lambda) f=\partial_{\vec n} (u(\lambda)|_{M\setminus M_\epsilon})-\partial_{\vec n} (u(\lambda)|_{M_\epsilon}),
 $$
 where $\vec n$ is the outward (for $M_\epsilon$) unit normal to  $\partial M_\epsilon\setminus\partial M$ with respect to $m_\varphi$;  note that there are no conical singularities of $m_\varphi$ on $\partial M_\epsilon$ since  $\epsilon>0$ is sufficiently small.  The operator $R^\varphi_\epsilon(\lambda)$ is an invertible first order  elliptic  classical pseudodifferential operator on $\partial M_\epsilon\setminus\partial M$. In particular,   on each component $\partial D_\epsilon^j$ of $\partial M_\epsilon\setminus\partial M$ the principal symbol of $R^\varphi_\epsilon(\lambda)$ is given by   $\sigma(x,\xi)=2\epsilon^{\beta_j}e^{\phi_j(x)}|\xi|$, which can be easily seen from the representation
 \begin{equation}\label{Aug16Fri}
 R^\varphi_\epsilon(\lambda)^{-1}=\Bigl((\Delta^\varphi+\lambda)^{-1}(\cdot\otimes\delta_{\partial M_\epsilon\setminus\partial M})\Bigr)|_{\partial M_\epsilon\setminus\partial M}, 
 \end{equation}
 where  $\delta_{\partial M_\epsilon\setminus\partial M}$ is the Dirac $\delta$-function along ${\partial M_\epsilon\setminus\partial M}$, the action of the resolvent is understood in the sense of distributions, and $\Delta^\varphi=-|x|^{2\beta_j}e^{-2\phi_j(x)}4\partial_x\partial_{\bar x}$ in the local parameter $x$ centred at $P^j$; cf.~\cite[Thm 2.1]{Carron} and~\cite[Sec. 4.4]{BFK}. As a consequence, for $s\in\Bbb C$, $\Re s>1$, the operator $R^\varphi_\epsilon(\lambda)^{-s}$ in $L^2(\partial M_\epsilon\setminus\partial M)$  is  trace class  and its zeta function $s\mapsto \zeta(s,\lambda)=\Tr R^\varphi_\epsilon(\lambda)^{-s}$ is holomorphic. Moreover,  $s\mapsto \zeta(s,\lambda)$ admits a meromorphic continuation from the half-plane $\Re s>1$ to $\Bbb C$ with no pole at $s=0$; see e.g.~\cite{Shubin}.  We set
 $$
 \det   R^\varphi_\epsilon(\lambda)=e^{-\partial_s \zeta (0,\lambda)}. 
 $$
\begin{lemma} \label{ThmA}The formula
$$\det(\Delta^\varphi+\lambda)=C\det(\Delta^\varphi_{\partial M_\epsilon}+\lambda) \det   R^\varphi_\epsilon(\lambda)$$
is valid, where $C$ is independent of $\lambda>0$. 
\end{lemma}
\begin{proof} 
The assertion is an analogue of~\cite[Theorem A]{BFK}.

The relation~\eqref{Aug16Fri} also implies that $\lambda\mapsto R^\varphi_\epsilon(\lambda)$ is an analytic family of pseudodifferential operators and  the order of   $\partial_\lambda^\ell R^\varphi_\epsilon(\lambda)^{-1}$ is $-1-2\ell$. Thus the order of $\partial_\lambda^\ell R^\varphi_\epsilon(\lambda)^{-1}$ is $1-2\ell$ and   $\bigl[\partial_\lambda R^\varphi_\epsilon(\lambda)\bigr]R^\varphi_\epsilon(\lambda)^{-1}$ is a trace class operator in $L^2(\partial M_\epsilon\setminus\partial M)$. As a consequence we have
\begin{equation}\label{aug16:02}
\partial_\lambda \log \det R^\varphi_\epsilon(\lambda)=\Tr\left(\bigl[\partial_\lambda R^\varphi_\epsilon(\lambda)\bigr]R^\varphi_\epsilon(\lambda)^{-1}\right),
\end{equation}
see~\cite[Prop. 1.1]{Forman}. 
 By writing the Schwartz kernel of  $\bigl(\partial_\lambda R^\varphi_\epsilon(\lambda)\bigr)R^\varphi_\epsilon(\lambda)^{-1}$  in terms of those of $(\Delta^\varphi +\lambda)^{-1}$ and $(\Delta^\varphi_{\partial M_\epsilon} +\lambda)^{-1}$ it is not hard to verify that
\begin{equation}\label{Aug16:03}
\Tr\left(\bigl[\partial_\lambda R^\varphi_\epsilon(\lambda)\bigr]R^\varphi_\epsilon(\lambda)^{-1}\right)=\Tr\left((\Delta^\varphi_{\partial M_\epsilon} +\lambda)^{-1}-(\Delta^\varphi +\lambda)^{-1}\right);
\end{equation}
the corresponding calculation can be found in~\cite[Proof of Thm 2.2]{Carron}, we omit the details.

It remains to show that 
\begin{equation}\label{AuG16:04}
\partial_\lambda \left[\log\det (\Delta^\varphi +\lambda)-\log\det (\Delta^\varphi_{ \partial M_\epsilon} +\lambda)\right]=\Tr\left((\Delta^\varphi_{\partial M_\epsilon} +\lambda)^{-1}-(\Delta^\varphi +\lambda)^{-1}\right);
\end{equation}
here we closely follow~\cite[Proof of Thm 4.2]{Carron}.
By the Krein theorem  (see e.g.~\cite[Ch. 8.9]{Yafaev}) there exists a spectral shift function $\xi\in L^1(\Bbb R_+,(1+\mu)^{-2}\,d\mu)$ such that 
$$
\Tr\left((\Delta^\varphi_{\partial M_\epsilon} +\lambda)^{-s}-(\Delta^\varphi +\lambda)^{-s}\right)=\int_0^\infty \xi(\mu)\frac{s \,d\mu}{(\mu+\lambda)^{s+1}},
$$
where $\Re s>1$ or $s=1$. 
In fact, in our case the spectrum of selfadjoint operator $\Delta^\varphi$ (resp. $ \Delta^\varphi_{ \partial M_\epsilon}$ ) in $L^2(M,m_\varphi)$ consists of  isolated eigenvalues $0\leq \lambda_1\leq \lambda_2\leq\lambda_3\leq \cdots$  and hence $\xi(\mu)=N_{\Delta^\varphi}(\mu)-N_{\Delta^\varphi_{ \partial M_\epsilon}}(\mu)$, where  $N_{\Delta^\varphi}(\mu)=\#\{k: \lambda_k<\mu, \lambda_k\in\sigma(\Delta^\varphi)\}$ is the spectral counting function of $\Delta^\varphi$ and $N_{\Delta^\varphi_{ \partial M_\epsilon}}(\mu)$ is the spectral counting function of $\Delta^\varphi_{ \partial M_\epsilon}$.  We have
$$
\partial_\lambda \left[\log\det (\Delta^\varphi +\lambda)-\log\det (\Delta^\varphi_{ \partial M_\epsilon} +\lambda)\right]=-\partial_\lambda \left[\partial_s\int_0^\infty \xi(\mu)\frac{s \,d\mu}{(\mu+\lambda)^{s+1}}\right]_{s=0}
$$
$$
=\int_0^\infty \xi(\mu)\frac{\,d\mu}{(\mu+\lambda)^{2}}=\Tr\left((\Delta^\varphi_{\partial M_\epsilon} +\lambda)^{-1}-(\Delta^\varphi +\lambda)^{-1}\right).
$$
This proves~\eqref{AuG16:04}.  Now the assertion of lemma follows from~\eqref{aug16:02},~\eqref{Aug16:03}, and~\eqref{AuG16:04}.
 \end{proof}

In the same way as before we define the Neumann jump operator $R^\varphi_\epsilon(\lambda)$ for $\lambda=0$ and denote $R^\varphi_\epsilon=R^\varphi_\epsilon(0)$. If $\partial M\neq \varnothing$, then the operators $\Delta^\varphi$ and $R^\varphi_\epsilon$ are still invertible and we define $\det \Delta^\varphi$ and $\det R^\varphi_\epsilon$ by setting $\lambda=0$ in the definitions for  $\det(\Delta^\varphi+\lambda)$ and $\det R^\varphi_\epsilon(\lambda)$. If $\partial M=\varnothing$, then  both  $\Delta^\varphi$ and $R^\varphi_\epsilon$ have zero as a simple eigenvalue (and the corresponding kernels consist of constant functions on $M$ and $\partial M_\epsilon$ respectively). In this case we introduce the modified determinant (i.e. with zero eigenvalue excluded). Namely, we set 
$\det \Delta^\varphi=e^{-\partial_s\zeta^*(0)}$, where for $\zeta^*(s)$ we may write
$$
\zeta^*(s)=\sum_{k:0<\lambda_k\in\sigma(\Delta^\varphi)}\lambda_k^{-s}=\frac 1 {\Gamma(s)}\int_0^\infty t^{s-1}\Bigl( \Tr \bigl(e^{-t\Delta^\varphi}\bigr)-1\Bigr)\, dt;
$$
similarly,  $\det  R^\varphi_\epsilon$ is defined via $\zeta^*(s)=\Tr (R^\varphi_\epsilon P^\bot)^{-s}$, where $ P^\bot$ is the orthogonal projection onto $(\ker R^\varphi_\epsilon)^\bot$ in  $L^2(\partial M_\epsilon, m_\varphi)$. 
 \begin{proposition}[BFK formulas]\label{BFKf}
1. If $\partial M=\varnothing$, then 
 \begin{equation}\label{Fr08:50}
\det  \Delta^\varphi =A_\varphi \det \Delta^\varphi_{\partial M_\epsilon} \frac {\det R^\varphi_\epsilon}{L_\varphi(\partial M_\epsilon)}, 
\end{equation}
where $A_\varphi$ is the total area of $M$ and $L_\varphi(\partial M_\epsilon)$ is the length of the boundary $\partial M_\epsilon$ in the metric $m_\varphi$.

2.  If $\partial M\neq\varnothing$ and $\partial M\cap\supp\pmb\beta=\varnothing$, then 
 \begin{equation}\label{Fr08:41}
\det \Delta^\varphi = \det \Delta^\varphi_{\partial M_\epsilon}     {\det R^\varphi_\epsilon}. 
\end{equation}
\end{proposition}  
The proof is preceded by

 \begin{lemma}[After L. Friedlander \& A. Voros] \label{FrVo} The functions $\lambda\mapsto \log\det ( \Delta^\varphi+\lambda)$ and 
 $
\lambda\mapsto  \log\det ( \Delta^\varphi_{\partial M_\epsilon}+\lambda)
 $
 admit  asymptotic expansions with zero constant terms as $\lambda \to +\infty$. 
  \end{lemma}
  \begin{proof} If   the metric $m_\varphi$ representing the divisor $\pmb\beta$ is flat in a neighborhood of $\supp\pmb \beta$, then there are no logarithms in the asymptotic expansion~\eqref{hta} and the assertion is due to Friedlander \& Voros ~\cite{Friedlander,Voros}. 
Here we adapt the Voros' argument, cf.~\cite[Prop. 2.7]{Lee BFK}.     
 
 Consider, for instance, the spectral zeta function $\zeta(s,\lambda)=\Tr(\Delta^\varphi+\lambda)^{-s}$, which is well defined for $\Re s>1$ and $\Re\lambda >0$.  Let  $\eta(s,\lambda)=\zeta(s,\lambda)\Gamma(s)$, then 
   \begin{equation}\label{eta}
  \eta(s,\lambda)=\int_0^\infty t^{s-1} e^{-\lambda t}\Tr e^{-t\Delta^\varphi}\,dt.
  \end{equation}
 For $\Re s>1$, $\eta(s,\lambda)$ can be expanded in $\lambda$ as $|\lambda|\to\infty$ by formally substituting the asymptotic expansion~\eqref{hta} into~\eqref{eta}. After the change of variable $t\mapsto  t/\lambda$ we get
  $$
   \begin{aligned}
  \eta(s,\lambda)\sim a_{-1} \lambda^{1-s} \int_0^\infty t^{s-2} e^{- t}\,dt+a_{-1/2}\lambda^{\frac 1 2-s} \int_0^\infty t^{s-\frac 3 2} e^{- t}\,dt+{a_0\lambda^{-s} \int_0^\infty t^{s-1} e^{- t}\,dt}
  \\
  +\sum_{j=1}^\infty\sum_{k=0}^{m_j}\sum_{\ell=0}^k a_{jk}\lambda^{-s-\frac j 2} \binom{k}{\ell}(-\log \lambda )^\ell\int_0^\infty t^{\frac j 2+s-1} e^{- t} \log^{k-\ell} t \,dt
  \\
  =a_{-1} \lambda^{1-s} \Gamma(s-1)+a_{-1/2}\lambda^{\frac 1 2-s} \Gamma\left (s-\frac 1 2\right)  +{a_0 \lambda^{-s}\Gamma(s)} 
  \\
  +\sum_{j=1}^\infty\sum_{k=0}^{m_j}\sum_{\ell=0}^k a_{jk}\lambda^{-s-\frac j 2} \binom{k}{\ell}(-\log \lambda)^\ell\Gamma^{(k-\ell)}\left(s+\frac j 2 \right).
  \end{aligned}
  $$
Thus
$$
 \begin{aligned}
\zeta(s,\lambda)\sim\frac {\lambda^{-s}}{\Gamma(s)}\Biggl(a_{-1}\lambda\Gamma(s-1)+a_{-1/2}\lambda^{\frac 1 2 } \Gamma\left (s-\frac 1 2\right)  +{a_0 \Gamma(s)}
\\
+ \sum_{j=1}^\infty\sum_{k=0}^{m_j}\sum_{\ell=0}^k a_{jk}\lambda^{-\frac j 2}\binom{k}{j}(-\log \lambda)^\ell\Gamma^{(k-\ell)}\left(s+\frac j 2 \right)\Biggr). 
\end{aligned}
$$  
  All functions involved are  meromorphic functions of $s$. Moreover, $s=0$ is a regular point of $\zeta(s,\lambda)$ and thus $\zeta'(0,\lambda)$ admits an asymptotic expansion in $\lambda$ of the form  
 $$
 \begin{aligned}
 \zeta'(0,\lambda)\sim a_{-1}\lambda(\log\lambda-1)-2a_{-1/2}\sqrt{\pi}\lambda^{\frac 1 2 }-{a_0 \log\lambda}
 \\
 +\sum_{j=1}^\infty\sum_{k=0}^{m_j}\sum_{\ell=0}^k a_{jk}\lambda^{-\frac j 2}\binom{k}{\ell}(-\log \lambda)^\ell\Gamma^{(k-\ell)}\left(\frac j 2 \right), 
 \end{aligned}
 $$
 where there is no constant term.
 \end{proof}

\begin{proof}[Proof of Proposition~\ref{BFKf}] Following~\cite{BFK}, we evaluate the constant $C$ in Lemma~\ref{ThmA} by considering the asymptotic expansion of all determinants involved in  
\begin{equation}\label{stRt}
\log \det (\Delta^\varphi +\lambda)=\log C+\log \det (\Delta^\varphi_{\partial M_\epsilon} +\lambda) +\log {\det R^\varphi_\epsilon(\lambda)}
\end{equation}
as $\lambda\to+\infty$.
It is known that the function
$
\lambda\mapsto\log {\det R^\varphi_\epsilon(\lambda)}
$
admits an asymptotic expansion with zero constant term; for the proof in the case $\partial M=\varnothing$ we refer to~\cite[Sec. 4.7]{BFK}, the case  $\partial M\neq\varnothing$ is studied in~\cite[Sec. 2 \& 3]{LeeBFKboundary}. For the other two determinants in~\eqref{stRt} we have proved the same fact in Lemma~\ref{FrVo}.  Thus we conclude that $C=1$ and hence 
\begin{equation}\label{Fr08:30}
\det (\Delta^\varphi +\lambda)=\det (\Delta^\varphi_{\partial M_\epsilon} +\lambda) {\det R^\varphi_\epsilon(\lambda)},\quad \lambda>0.
\end{equation}
It remains to pass in~\eqref{Fr08:30} to the limit as $\lambda\to0+$. 

 In the case $\partial M=\varnothing$ only  $\Delta^\varphi_{\partial M_\epsilon}$ is positive and thus $ \det (\Delta^\varphi_{\partial M_\epsilon} +\lambda)\to  \det \Delta^\varphi_{\partial M_\epsilon}$ as $\lambda\to 0+$.  From the definition of the modified determinant it immediately follows that 
$$
\log \det (\Delta^\varphi +\lambda)=\log\lambda+\log \det \Delta^\varphi+o(1),\quad \text{ as } \lambda\to 0+.  
$$  
Clearly, $\Delta^\varphi A_\varphi^{-1/2}=0$ and $\|A_\varphi^{-1/2}\|_{L^2(M,m_\varphi)}=1$; recall that $A_\varphi$ stands for  the total area of $M$ in the metric $m_\varphi$.  Hence for any $F\in L^2(M,m_\varphi)$ we have 
 \begin{equation}\label{AddnSun}
  (\Delta^\varphi+\lambda)^{-1}F=\frac 1 {A_\varphi\lambda}(F,1)_{L^2(M,m_\varphi)}+ (\Delta^\varphi+\lambda)^{-1}\left (F-\frac 1 A_\varphi(F,1)_{L^2(M,m_\varphi)}\right),
 \end{equation}
 where the second term in the right hand side is holomorphic in $\lambda$, $|\lambda|\ll 1$. The relation~\eqref{Aug16Fri} implies that for $\lambda\geq 0$ the operator $R^\varphi_\epsilon(\lambda)$ in $L^2(\partial M_\epsilon,m_\varphi)$ is selfadjoint and nonnegative, together with~\eqref{AddnSun} it also gives
$$
R^\varphi_\epsilon(\lambda)^{-1} =\frac 1 {\lambda A_\varphi} (\cdot, 1)_{L^2(\partial M_\epsilon,m_\varphi)}+\mathfrak h(\lambda),
$$
where $\|\mathfrak h(\lambda)\|_{\mathcal B (L^2(\partial M_\epsilon,m_\varphi))}=O(1)$ as $\lambda\to0$. Therefore, as $\lambda\to 0+$ the first eigenvalue $\mu_0(\lambda)=1/\|R^\varphi_\epsilon(\lambda)^{-1}\|_{\mathcal B(L^2(\partial M_\epsilon,m_\varphi))}$ of $R^\varphi_\epsilon(\lambda)$ goes to zero, while the others satisfy $\mu_k(\lambda)\geq \delta$ with some $\delta>0$. 
Finally, for the determinant of $R^\varphi_\epsilon(\lambda)$ we obtain 
$$
\log \det R^\varphi_\epsilon(\lambda)=\log \mu_0(\lambda)+\log \det R^\varphi_\epsilon+o(1)$$
$$=\log{\frac { A_\varphi}{L_\varphi(\partial M_\epsilon)}}+\log \lambda+\log \det R^\varphi_\epsilon+o(1)\quad \text{ as } \lambda\to 0+;
$$
here  $L_\varphi(\partial M_\epsilon)$ is the norm of  the operator $(\cdot, 1)_{L^2(\partial M_\epsilon,m_\varphi)}$ in the space of bounded operators acting in $L^2(\partial M_\epsilon,m_\varphi)$.
Thus passing in~\eqref{Fr08:30} to the limit we get~\eqref{Fr08:50}.

In the case $\partial M\neq \varnothing$ the operators $\Delta^\varphi $, $\Delta^\varphi_{\partial M_\epsilon}$, and $R^\varphi_\epsilon$ are positive and hence the determinants in~\eqref{Fr08:30} tend to the corresponding determinants in~\eqref{Fr08:41} as $\lambda\to0+$.
\end{proof}

\subsection{Proof of Theorem~\ref{main} and Corollary~\ref{2conical}}\label{Pr}
Recall that $m_\varphi=e^{2\varphi}m_0$, where $m_\varphi$ is a conical  metric representing a divisor $\pmb \beta$ and $m_0$ is a smooth conformal metric on $M$. In addition to the BFK decomposition formulas obtained in Proposition~\ref{BFKf} we will also be using similar decomposition formulas for $\det \Delta^0$. The  latter formulas  can be formally obtained by setting $\varphi=0$ in~\eqref{Fr08:50},~\eqref{Fr08:41}, and the definitions for $\Delta^\varphi_{\partial M_\epsilon}$, ${\det R^\varphi_\epsilon}$, and ${L_\varphi(\partial M_\epsilon)}$ in Section~\ref{sBFK}. We only notice that the corresponding results are well known: since the metric $m_0$ is smooth, the formula~\eqref{Fr08:50} (resp.~\eqref{Fr08:41}) with $\varphi=0$ is a particular case of~\cite[Theorem $\rm B^*$]{BFK} (resp.~\cite[Corollary 1.3]{LeeBFKboundary}). 

\begin{lemma}\label{confin} Let  $\epsilon>0$ be sufficiently small.  
\begin{enumerate}
\item If $\partial M=\varnothing$, then $\frac {\det R^\varphi_\epsilon}{L_\varphi(\partial M_\epsilon)}=\frac {\det R^0_\epsilon}{L_0(\partial M_\epsilon)}$.
\item If $\partial M\neq\varnothing$ and $\partial M\cap\supp\pmb\beta=\varnothing$, then
${\det R^\varphi_\epsilon}={\det R^0_\epsilon}$. 
\end{enumerate}
\end{lemma}
\begin{proof} For all sufficiently small $\epsilon>0$ the disks $D^1_\epsilon,\dots,D^n_\epsilon$ are disjoint and do not touch the boundary $\partial M$.  In each disk $D^j_\epsilon$ we replace $\varphi(x)=\beta_j\log|x|+\phi_j(x)-\psi_j(x)$ by the smooth potential  
$$
\tilde \varphi(x)=\chi(|x|/\epsilon)\bigl(\beta_j\log|x|+\phi_j(x)\bigr)-\psi_j(x),
$$
where $\chi\in C^\infty (\Bbb R_+)$ is a cutoff function with properties:  $\chi(|x|)=0$ for $|x|\leq 1/3$ and $\chi(|x|)=1$ for $|x|\geq1/2$. We also set $\tilde\varphi=\varphi$ on $M_\epsilon$. As a result we obtain $\tilde\varphi\in C^\infty(M)$ such that  $\varphi=\tilde\varphi$ on $M_{\epsilon/2}\supset M_\epsilon$. Hence  $L_{\tilde\varphi}(\partial M_\epsilon)=L_{\varphi}(\partial M_\epsilon)$ and $R^{\tilde\varphi}_\epsilon=R^{\varphi}_\epsilon$ (recall that $\Delta^\varphi_{\partial M_\epsilon}$ is the Friederichs extension of the Dirichlet Laplacian and hence the solution $u=\hat f-(\Delta^\varphi_{\partial M_\epsilon})^{-1} \Delta^\varphi \hat f$  to~\eqref{BVP} with $\lambda=0$ is bounded and  thus coincides with $\tilde u=\hat f-(\Delta^{\tilde\varphi}_{\partial M_\epsilon})^{-1}\Delta^{\tilde\varphi}\hat f$, where $\Delta^{\tilde\varphi}_{\partial M_\epsilon}$ is the selfadjoint Dirichlet Laplacian).

 As is known~\cite{Wentworth}, the invariance of $\frac {\det R^0_\epsilon}{L_0(\partial M_\epsilon)}$  in the case $\partial M=\varnothing$ (resp. of $\det R^0_\epsilon$ in the case $\partial M\neq\varnothing$) under the conformal transformations $m_0\mapsto e^{2\tilde\varphi}m_0$ with smooth $\tilde\varphi$ can be easily seen from the BFK formula~\eqref{Fr08:50} (resp.~\eqref{Fr08:41}) together with Polyakov/Polyakov-Alvarez formulas  on $M$, $M_\epsilon$, and $D^j_\epsilon$.  
\end{proof}

\begin{proof}[{Proof of Theorem~\ref{main}}]
BFK formulas in Proposition~\ref{BFKf} and Lemma~\ref{confin} imply
\begin{equation}\label{AUG20}\begin{aligned}
\log \frac {(\det \Delta^\varphi)/A_\varphi }{(\det \Delta^0)/A_0}=\log \frac{ \det \Delta^\varphi_{\partial M_\epsilon} } {\det \Delta^0_{\partial M_\epsilon} }\text{ if } \partial M=\varnothing,
\\
\log \frac {\det \Delta^\varphi }{\det \Delta^0}=
\log \frac{ \det \Delta^\varphi_{\partial M_\epsilon} } {\det \Delta^0_{\partial M_\epsilon} }\text{ if } \partial M\neq \varnothing .
\end{aligned}
\end{equation}
Note that $\Delta^\varphi_{\partial M_\epsilon} $ can be decomposed into the direct sum of operators: 
\begin{equation}\label{DS}
\Delta^\varphi_{\partial M_\epsilon} = \Delta^\varphi_{M_\epsilon}\oplus_{j=1}^n  \Delta^\varphi_{D^j_\epsilon},
\end{equation}
where $\Delta^\varphi_{M_\epsilon}$ is the selfadjoint  Dirichlet Laplacian on $(M_\epsilon, m_\varphi)$ and $\Delta^\varphi_{D^j_\epsilon}$ is the Friederichs extension of the Dirichlet Laplacian on the metric disk $(D^j_\epsilon, |x|^{2\beta_j}e^{2\phi_j}|dx|^2)$ studied in Section~\ref{SSnon-flat}.
As a consequence we have
\begin{equation}\label{DSD}
\det \Delta^\varphi_{\partial M_\epsilon} =\det \Delta^\varphi_{M_\epsilon}\prod_{j=1}^n  \det\Delta^\varphi_{D^j_\epsilon}.
\end{equation}
Similarly we decompose $\Delta^0_{\partial M_\epsilon}$ into the corresponding direct sum and obtain~\eqref{DS} and~\eqref{DSD} with $\varphi$ replaced by $0$.
This together with Lemma~\ref{KeyLemma}, formula~\eqref{zeta at zero+} for $\zeta_<(0,\beta)$, and~Lemma~\ref{SimpleKey} implies 
\begin{equation}\label{E6s}
\begin{aligned}
\log \frac{ \det \Delta^\varphi_{\partial M_\epsilon} } {\det \Delta^0_{\partial M_\epsilon} }=\log\frac{ \det \Delta^\varphi_{M_\epsilon}}{\det \Delta^0_{M_\epsilon}}-\sum_{j=1}^n \left(\frac 1 6 (\beta_j^2+2\beta_j)\log\epsilon+2\phi_j(0)\,\zeta_<(0,\beta_j)\right.
\\
\left.-\frac 1 3 \psi_j(0)+C(\beta_j)+\frac {\beta_j} 2 \right) +o(1),\quad \epsilon\to 0+,
\end{aligned}
\end{equation}
where we introduced the notation
$$
C(\beta)=\zeta'_<(0,\beta)-\left(2 \zeta_<(0,\beta)-\frac 1 3\right)\log 2  -\zeta_<'(0,0)- \frac \beta 2.
$$ 
The formula~\eqref{Cbeta} for $C(\beta)$ now follows from~\eqref{zeta at zero+},~\eqref{Spr+}, and~Lemma~\ref{Barnes} in Appendix~\ref{BarnesPV}.

The classical Polyakov-Alvarez formula on $M_\epsilon$ reads
\begin{equation}\label{E5s}
\begin{aligned}
\log\frac{ \det \Delta^\varphi_{M_\epsilon}}{\det \Delta^0_{M_\epsilon}}=-\frac{1}{6\pi}\left (\frac1 2 \int_{M_\epsilon} (|\nabla_0\varphi|^2+2K_0\varphi)\,dA_0+\int_{\partial M_\epsilon}k_0\varphi \,d s_0  \,\right )\\-\frac 1 {4\pi}\int_{\partial M_\epsilon}{\partial_{\vec n}}\varphi\,d s_0;
\end{aligned}
\end{equation}
see e.g.~\cite{Alvarez,OPS,Weisberger}.
Let us rewrite the right hand side of~\eqref{E5s} in the form
\begin{equation}\label{ee}
\begin{aligned}
-\frac{1}{6\pi}\left (\frac1 2 \int_{M_\epsilon} (\varphi\Delta^0\varphi+2K_0\varphi)\,dA_0+\frac 1 2 \int_{\partial M_\epsilon} \varphi \partial_{\vec n}\varphi \,ds_0+\int_{\partial M_\epsilon}k_0\varphi \,d s_0  \,\right )\\-\frac 1 {4\pi}\int_{\partial M_\epsilon}\partial_{\vec n}\varphi\,d s_0. 
\end{aligned}
\end{equation}

In the disk  $ D^j_\epsilon$ we  have $m_\varphi=|x|^{2\beta_j}e^{2\phi_j}|dx|^2$,  $m_0=e^{2\psi_j}|dx|^2$, and 
$$\varphi(x)=\beta_j\log|x|+\phi_j(x)-\psi_j(x).$$  

Taking into account the equalities
\begin{equation*}
k_0=-e^{-\psi}(\epsilon^{-1}+\partial_{|x|}\psi_j), z=\epsilon e^{i\theta},  ds_0=e^{\psi}\epsilon\, d\theta,
\end{equation*}
on $\partial D^j_\epsilon$,  
for the integrals along the $j$-th  component $\partial D^j_\epsilon$ of the boundary $\partial M_\epsilon\setminus\partial M$  in~\eqref{ee} we get 
$$
\begin{aligned}
\int_{\partial D^j_\epsilon}\partial_{\vec n}\varphi \,ds_0=-\int_0^{2\pi}e^{-\psi_j(x)}\partial_{|x|}\bigl(\beta_j\log|x|+\phi_j(x)-\psi_j(x)\bigr)e^{\psi_j(x)}\Bigr|_{x=\epsilon e^{i\theta}}\epsilon\, d\theta
\\
=-\int_0^{2\pi}\left(\frac {\beta_j} \epsilon+O(1)\right)\epsilon\, d\theta=-2\pi\beta_j+O(\epsilon),
\end{aligned}
$$

$$
\begin{aligned}
\int_{\partial D^j_\epsilon}  & k_0\varphi \,d s_0 
\\ & =-\int_0^{2\pi}e^{-\psi_j(x)}\bigl(\epsilon^{-1}+\partial_{|x|}\psi_j(x)\bigr)\bigl(\beta_j\log|x|+\phi_j(x)-\psi_j(x)\bigr)e^{\psi_j(x)}\Bigr|_{x=\epsilon e^{i\theta}}\epsilon\, d\theta
\\
& =-2\pi\beta_j\log\epsilon-2\pi\bigl(\phi_j(0)-\psi_j(0)\bigr)+o(1),
\end{aligned}
$$
$$
\begin{aligned}
\int_{\partial D^j_\epsilon} \varphi \partial_{\vec n}\varphi \,ds_0
&= -\int_0^{2\pi}\bigl (\beta_j\log|x|+\phi_j(x)-\psi_j(x)\bigr)e^{-\psi_j(x)}
\\ &\hspace{3cm} \times \partial_{|x|}\bigl(\beta_j\log|x|+\phi_j(x)-\psi_j(x)\bigr)e^{\psi_j(x)}\Bigr|_{x=\epsilon e^{i\theta}}\epsilon\, d\theta
\\
& =-\int_0^{2\pi}\bigl(\beta_j\log|x|+\phi_j(x)-\psi_j(x)\bigr)\Bigr|_{x=\epsilon e^{i\theta}}
\left(\frac{\beta_j}\epsilon+O(1)\right)\epsilon\, d\theta
\\
& =-2\pi\beta_j^2\log\epsilon-2\pi\beta_j\bigl(\phi_j(0)-\psi_j(0)\bigr)+o(1). 
\end{aligned}
$$
In~\eqref{ee} we also use the identities $K_\varphi=e^{-2\varphi}(K_0+\Delta^0\varphi)$ on $M_\epsilon$ and $dA_\varphi=e^{2\varphi}dA_0$ and  finally obtain from~\eqref{E5s} the following: 
\begin{equation*}
\begin{aligned}
\log\frac{ \det \Delta^\varphi_{M_\epsilon}}{\det \Delta^0_{M_\epsilon}}=&-\frac{1}{12\pi}\left(\int_M K_\varphi \varphi\,dA_\varphi+\int_M K_0\varphi\,dA_0 +\int_{\partial M} \varphi\partial_{\vec n}\varphi\,ds_0\right)
\\
& -\frac 1{6\pi}\int_{\partial M}k_0\varphi\, ds_0-\frac 1 {4\pi}\int_{\partial M} \partial_{\vec n}\varphi\,ds_0\\
 & +\sum_{j=1}^n \left( \frac 1 6 (\beta_j^2+2\beta_j)\log\epsilon+\frac {\beta_j+2} {6} \bigl(\phi_j(0)-\psi_j(0)\bigr)+\frac {\beta_j} 2\right) +o(1).
\end{aligned}
\end{equation*}
This together with~\eqref{E6s} and~\eqref{AUG20} implies the desired formulas~\eqref{Pol} and~\eqref{Pol-Alv} (if $\partial M=\varnothing$, then the integrals along $\partial M$ do not appear).
\end{proof}

\begin{proof}[Proof of Corollary~\ref{2conical}] Consider for instance the case $\partial M=\varnothing$. Let us change the notation: by $m_\alpha$ and $m_\beta$  we denote the metrics $m_0$ and $m_\varphi$ representing the divisors $\pmb\alpha$ and $\pmb \beta$ respectively. Let also $m_0$ be a smooth metric in the same conformal class,  $m_\alpha=e^{2\alpha}m_0$ and $m_\beta=e^{2\beta}m_0$ with some functions $\alpha\in C^\infty(M\setminus\supp\pmb\alpha)$ and $\beta\in C^\infty(M\setminus\supp\pmb\beta)$, then $\varphi=\beta-\alpha$ ($m_0$ can be constructed by smoothing a metric potential  in the same way as in the proof of Lemma~\ref{confin}). In a local parameter centered at $P^j\in\supp\pmb\alpha\cup\supp\pmb\beta$ we write $m_\alpha=|x|^{2\alpha_j}e^{2\hat\alpha_j(x)}|dx|^2$, $m_\beta=|x|^{2\beta_j}e^{2\hat\beta_j(x)}|dx|^2$, and $m_0=e^{2\psi_j(x)}|dx|^2$.
In the same way as in the proof of Theorem~\ref{main}  the integral over $M$ below reduces to a sum of line integrals with contours shrinking to the conical singularities of $m_\alpha$ and $m_\beta$. As a result we obtain
$$
-\frac 1 {12\pi}\int_{M} [(\Delta^0\alpha) \beta -\alpha(\Delta^0\beta)]\,dA_0=\frac  1 6 {\sum_{k=1}^n\alpha_k}\bigl(\psi_k(0)-\hat\beta_k(0)\bigr) -\frac  1 6 {\sum_{j=1}^{n}\beta_j}\bigl(\psi_j(0)-\hat\alpha_j(0)\bigr). $$
This together with identity $K_0=e^{2\varphi}\left(K_\varphi-\Delta^\varphi\varphi\right)$ and the Polyakov formula~\eqref{Pol} for $\varphi=\alpha$  and $\varphi=\beta$ leads to~\eqref{18:09}. The case $\partial M\neq\varnothing$ is similar.
\end{proof}

\section{Examples and Applications}\label{appl}

In this section we  illustrate our results by obtaining  new and recovering   known  explicit formulas for determinants of Laplacians on singular
surfaces with or without boundary. 

\subsection{Constant curvature spheres with two conical singularities}\label{AF}

By the uniformization theorem a Riemann surface with two conical singularities homeomorphic to a sphere is conformally equivalent to the Riemann sphere $\Bbb CP^1$;  we can assume that the conical singularities are at $z=0$ and $z=\infty$.  Let $m_\varphi$ be a corresponding  conformal metric on  $\Bbb CP^1$ with constant curvature $K_\varphi$. Then, by~\cite[Theorem II]{Troyanov Spheres},  $K_\varphi$ is positive and there exist  $\mu\in[0,\infty)$ and $\beta>-1$, such that either $\beta$ is an integer or $\mu=0$, and (up to  a change of coordinates $z\mapsto pz$ with a constant $p\in \Bbb C$) we have  
\begin{equation}\label{TroM}
m_\varphi=\frac{(2\beta+2)^2|z|^{2\beta}|dz|^2}{(|1+\mu z^{\beta+1}|^2+K_\varphi|z|^{2\beta+2})^2}. 
\end{equation}
This metric represents the divisor $\pmb \beta=\beta\cdot 0+\beta\cdot\infty$. 
The distance between $z=0$ and $z=\infty$ in the metric~\eqref{TroM} is  $d=\frac 2{\sqrt{K_\varphi}}\arctan\left(\sqrt{K_\varphi}/\mu\right)$. If $\beta\notin\Bbb N$, then $\mu=0$,  $d=\pi/\sqrt{K_\varphi}$, and the two conical singularities are  antipodal. 
\begin{proposition}\label{Klev} Let $K_\varphi>0$,  $\mu\in[0,\infty)$, and $\beta>-1$.  Then for the determinant of (the Friederichs extension of) the Laplacian $\Delta^\varphi$ on the Riemann sphere $\Bbb CP^1$ endowed with metric~\eqref{TroM} we have 
\begin{equation}\label{zetaPrimeSphere}
\begin{aligned}
 \log\det{\Delta^\varphi}  = &  -\frac 1 6\left(\beta+1-\frac 1 {\beta+1}\right)\log\left(1+\frac{\mu^2}{K_\varphi}\right)+\frac {\beta+1} 2\\ &-\frac 1 3\left(\beta+1+\frac 1 {\beta+1}\right)\log\frac{\beta+1}{\sqrt{K_\varphi}}-4\zeta'_B(0;\beta+1,1,1)-\log K_\varphi,
 \end{aligned}
 \end{equation}
 where either $\beta$ is an integer or $\mu=0$.
 \end{proposition}

The proof is preceded by  some remarks regarding the novelty of the result.  The case $\mu=0$ (a sphere with two antipodal conical singularities of angle $2\pi(\beta+1)$, or, equivalently, a spindle or an american football) was previously studied in~\cite{Spreafico3} by a different  approach based on separation of variables, an erratum was recently announced  in~\cite{Klevtsov}. For some closely related results in the study of Quantum Hall Effect on singular surfaces we refer to~\cite{CCLW,C-W,Klevtsov}. In its full generality the formula~\eqref{zetaPrimeSphere} appears for the first time.
A variational formula for $\log\det{\Delta^\varphi} $ with respect to $\mu$ (for $\beta=K_\varphi=1$) was  obtained in~\cite{KalvinJGA,IMRN}.
 In the case $\beta=K_\varphi=1$ the formula~\eqref{zetaPrimeSphere} simplifies to 
\begin{equation}\label{kokot}
\det\Delta^\varphi=\sqrt[6]{2}e^{1-2\zeta_R'(-1)}(1+\mu^2)^{-1/4};
\end{equation}
see~Lemma~\ref{Barnes} in Appendix~\ref{BarnesPV} (in order to compare~\eqref{kokot} with the result in~\cite{KalvinJGA,IMRN}, set  there $z_1=0$ and $z_2=1/\mu$). 
Thus we find that the undetermined in~\cite[formula (1.2) with $\rho(z,\bar z)=4(1+|z|^2)^{-2}$]{KalvinJGA} and~\cite[formula (1)]{IMRN} constant $C$  equals $\sqrt[6]{2}e^{1-2\zeta_R'(-1)}$. 
\begin{proof}[Proof of Proposition~\ref{Klev}]
 We will rely on~\eqref{Pol}, where as the reference metric $m_0$ we take the standard round curvature one metric on $\Bbb CP^1$, i.e.  $m_0=e^{2\psi}|dz|^2$ with 
 $$\psi(z)=\log2-\log(1+|z|^2).$$
  Consider the map 
$$
w=f(z)=\frac {z^{\beta+1}}{1+\mu z^{\beta+1}}: \Bbb CP^1_z\to \Bbb CP^1_w.
$$
It is a ramified covering with ramification divisor $\beta\cdot 0 + \beta\cdot \infty$. 
In the case $K_\varphi=1$ we have 
$$
m_\varphi=f^*\left(e^{2\psi(w)}|dw|^2\right)=|z|^{2\beta}e^{\phi(z)}|dz|^2,\quad 
\phi(z)=\psi\circ f(z)+\log|z^{-\beta}f'(z) |,
$$
and $A_\varphi=4\pi(\beta+1)$. The integral part of Polyakov type formula~\eqref{Pol} takes the form
\begin{equation*}
\begin{aligned}
\int_{\Bbb C}\varphi|z|^{2\beta}e^{2\phi}\frac{dz\wedge d\bar z}{-2i} & +\int_{\Bbb C}\varphi e^{2\psi}\frac{dz\wedge d\bar z}{-2i}
\\
=\int_{\Bbb C}\Bigl(\beta\log|z| & +\psi\circ f+\log|z^{-\beta}f'| -\psi\Bigr)|z|^{2\beta}e^{2\phi}\frac{dz\wedge d\bar z}{-2i}\\ &+\int_{\Bbb C}\Bigl(\beta\log|z| +\phi-\psi\Bigr) e^{2\psi}\frac{dz\wedge d\bar z}{-2i}. 
\end{aligned}
\end{equation*}
We remove the parenthesis and calculate the integrals term by term. For the first one we use the Liouville equation $|z|^{2\beta}e^{2\phi}=-4\partial_z\partial_{\bar z}\phi$ and then Green's theorem to get
\begin{equation}\label{eqnJul25}
\begin{aligned}
\int_{\Bbb C} & \bigl(\log|z|\bigr) |z|^{2\beta}e^{2\phi}\frac{dz\wedge d\bar z}{-2i}\\& =\lim_{\epsilon\to0+}\left(\oint_{|z|=1/\epsilon}+\oint_{|z|=\epsilon}\right)\Bigl(\log|z|\partial_{\vec n}\phi(z)-\phi(z)\partial_{\vec n}\log|z|\Bigr)|dz|
\\ & = 2\pi \Bigl(\log\frac {2\beta+2}{1+|\mu|^2}-\log(2\beta+2)  \Bigr)=- 2\pi\log(1+|\mu|^2). 
\end{aligned}
\end{equation}

For the second term we first observe that 
\begin{equation}\label{intformula2}
\int_{\Bbb C}\psi e^{2\psi}\frac{dz\wedge d\bar z}{-2i}=4\pi(\log 2-1)
\end{equation}
(e.g. by passing to the polar coordinates $(r,\theta)=(|z|,\arg z)$) and then by changing the variable $z\mapsto f(z)$ we obtain 
$$
 \int_{\Bbb C} (\psi\circ f) |z|^{2\beta} e^{2\phi} \frac{dz\wedge d\bar z}{-2i}=4\pi(\log 2 -1)(\beta+1).
$$

In the third term we represent $|z^{-\beta}f'|$ as $|f|^2(\beta+1)|z|^{-2\beta-2}$  and get
$$ \begin{aligned}
\int_{\Bbb C}\Bigl(\log|z^{-\beta}f'| \Bigr)|z|^{2\beta}e^{2\phi}\frac{dz\wedge d\bar z}{-2i}=2\int_{\Bbb C}\Bigl(\log|f| \Bigr)|z|^{2\beta}e^{2\phi}\frac{dz\wedge d\bar z}{-2i}\\+4\pi(\beta+1)\log(\beta+1)+4\pi(\beta+1)\log(1+|\mu|^2),
\end{aligned}
$$
where the integral in the right hand side is zero (as it follows e.g. from~\eqref{eqnJul25} with $\mu=\beta=0$  after  the change of variables $z\mapsto f(z)$). 

Next we use the Liouville equation for $\phi$ and $\psi$ and Green's formula  to evaluate fourth and sixth terms together. We have
$$
\begin{aligned}
 -& \int_{\Bbb C}\psi |z|^{2\beta}e^{2\phi}\frac{dz\wedge d\bar z}{-2i}+\int_{\Bbb C}\phi e^{2\psi}\frac{dz\wedge d\bar z}{-2i}
 =-\lim_{\epsilon\to0+} \oint_{|z|=1/\epsilon}\Bigl(\psi\partial_{\vec n}\phi-\phi\partial_{\vec n}\psi\Bigr)|dz| 
 \\
 &=-2\pi\left(2(\beta+1)\log 2 -2\log\frac {2\beta+2}{1+|\mu|^2}\right).
\end{aligned}
$$
The fifth term can be integrated as in~\eqref{eqnJul25} and gives zero. For the seventh term see~\eqref{intformula2}.  
 
In total for the integral part we  get
\begin{equation*}\label{1thu25}
\int_{\Bbb C}\varphi|z|^{2\beta}e^{2\phi}\frac{dz\wedge d\bar z}{-2i}+\int_{\Bbb C}\varphi e^{2\psi}\frac{dz\wedge d\bar z}{-2i}
= 2\pi\beta\log(1+|\mu|^2)-4\pi \beta+4\pi(\beta+2)\log(\beta+1).
\end{equation*}
For the non-integral part of~\eqref{Pol}  we obtain
\begin{equation*}\label{2thu25}
\frac  1 6 {\sum_{j=1}^2 \beta}\left(\frac{\phi_j(0)}{\beta+1}-\psi_j(0)\right)- 2 C(\beta)
=\frac  1 6 { \beta}\left(\frac{\log(2\beta+2)}{\beta+1}+ \frac{\log\frac{2\beta+2}{1+|\mu|^2}}{\beta+1}-2\log 2 \right)- 2 C(\beta). 
\end{equation*}
It remains to notice that for the standard (smooth) curvature one sphere of area $A_0=4\pi$ one has 
\begin{equation}\label{WeisSphere}\log\det\Delta_0=1/2-4\zeta_R'(-1);
\end{equation}
see e.g.~\cite[p. 204]{OPS}.  This together with \eqref{Pol} and~\eqref{Cbeta} leads to~\eqref{zetaPrimeSphere} with $K_\varphi=1$.  

 In order to include into consideration the case $0<K_\varphi\neq 1$ we do the change of variables  $z\mapsto (K_\varphi)^{\frac 1 {2\beta+2}} z$ in the curvature one metric $m_\varphi$,  then divide the resulting metric by $K_\varphi$. In accordance with the rescaling argument  this decreases the value of $\zeta'(0)$ by $\zeta(0) \log K_\varphi$. It remains to note that for a sphere with two conical singularities of order  $\beta$  Corollary~\ref{z0} gives 
$
\zeta(0)  = \frac 1 6 \left(\beta+1+\frac 1{\beta+1}\right)-1
$.
\end{proof}

In the remaining part of this section we study extremal properties of the determinant $\det\Delta^\varphi$ on the constant curvature spheres  $(\Bbb CP^1,m_\varphi)$  with two conical singularities as a function of $\mu$ and $\beta$ while the area $A_\varphi=4\pi$  remains  fixed.  The Gauss-Bonnet theorem~\cite{Troyanov Curv} reads $K_\varphi=\beta+1$ and  from~\eqref{zetaPrimeSphere} we immediately obtain 
\begin{equation}\label{19:15}\begin{aligned}
\log & \det \Delta^{\varphi}_{\text{Area } 4\pi}= - \frac 1 6\left(\beta+1-\frac 1 {\beta+1}\right)\log\left(1+\frac{|\mu|^2}{\beta+1}\right)
\\& - \left(1+\frac 1 6\left(\beta+1+\frac 1 {\beta+1}\right)\right)\log (\beta+1)-4\zeta_B'(0;\beta+1,1,1)+\frac {\beta+1} 2.
\end{aligned}
\end{equation}

Clearly $\det \Delta^{\varphi}_{\text{Area } 4\pi}$ monotonically goes to zero as $|\mu|$ increases and $\beta\in\Bbb N$ remains fixed. 
If $\mu=0$ and  $\beta\to -1^+$ (or $\mu\in[0,\infty)$ and $\beta\to +\infty$), then  the value of $\det \Delta^{\varphi}_{\text{Area } 4\pi}$ increases without any bound (this can be easily seen from~\eqref{19:15} and available asymptotic expansions of the Barnes zeta function, see e.g.~\cite{Matsumoto},~\cite{Spreafico2} or~\cite[A.6]{Klevtsov}). Namely, we have
\begin{equation}\label{ASE}
\begin{aligned}
\log \det \Delta^{\varphi}_{\text{Area } 4\pi}=-\frac 1 {6(\beta+1)}\log(\beta+1)& -\frac 1 {\beta+1}\left(\frac 1 3 -4\zeta'_R(-1)\right)-\log \frac {\beta+1} {2\pi} 
\\
&-\frac 1 6 (\beta+1)\log(\beta+1) +O(\beta+1)\text{ as } \beta\to -1^+,\\
\log \det \Delta^{\varphi}_{\text{Area } 4\pi}=- \frac 1 6\left(\beta+1-\frac 1 {\beta+1}\right)& \log\left(1+\frac{|\mu|^2}{\beta+1}\right)\\
+\frac 1 {6}\left(\beta+1+\frac 1 {\beta+1}\right)& \log (\beta+1)+\left(\frac 1 6+ 4\zeta'_R(-1)\right)(\beta+1)
\\ &\hspace{2.4cm}+ \log {2\pi}+ O(1/\beta)\text{ as } \beta\to+\infty.
\end{aligned}
\end{equation}
In particular this demonstrates that  the  well-known inequality $$\det \Delta^{\varphi}_{\text{Area } 4\pi}\leq \exp(1/2-4\zeta_R'(-1))$$ valid for all smooth metrics $m_\varphi$ on $\Bbb CP^1$ of area $4\pi$ with equality iff $m_\varphi=4(1+|z|^2)^{-2}|dz|^2$ (i.e. iff  $(\Bbb CP^1,m_\varphi)$  is  isometric to the  standard round curvature one sphere $x_1^2+x_2^2+x_3^2=1$ in $\Bbb R^3$, cf.~\eqref{WeisSphere},~\cite[Corollary 1.(a)]{OPS}),  is no longer true even for constant positive curvature metrics with conical singularities.  It is interesting to note however that the case $\beta=0$ corresponds to the  standard round curvature one sphere  and provides $\det \Delta^{\varphi}_{\text{Area } 4\pi}$ from~\eqref{19:15} with the local maximum $\det \Delta^{\varphi}_{\text{Area } 4\pi}\Bigr |_{\beta=0}= \exp(1/2-4\zeta_R'(-1))$. Indeed, we have

\begin{proposition}\label{extremal} The determinant of Friederichs Laplacian on the metrics of constant curvature on a  sphere with two conical singularities and normalized area  $4\pi$ is unbounded from above and  attains its local maximum on  the metric of standard round sphere. Moreover, the following asymptotic expansion is valid
$$
 \begin{aligned}
\log \det \Delta^{\varphi}_{\text{Area } 4\pi}=\frac 1 2-4\zeta'_R(-1) - \left(\frac \gamma 3+\frac 1 9\right)\beta^2+\left(\frac \gamma 3+\frac 7 {36}\right)\beta^3+O\bigl(\beta^4\bigr)\text{ as } \beta\to 0,
\end{aligned}
$$
where  $\gamma=-\Gamma'(1)$ and $\beta$ is the order (i.e. $2\pi(\beta+1)$ is the angle) of the antipodal conical singularities;  
 see also Fig.~\ref{F1} for a graph of $\det \Delta^{\varphi}_{\text{Area } 4\pi}$. 
 \end{proposition}
\begin{proof} The expansion follows from~\eqref{19:15} (with $\mu=0$) and  the asymptotic expansion
$$\begin{aligned}
\zeta_B'(0;\beta+1,1,1)=\zeta'_R(-1) -\frac {5}{24}\beta+ \left( \frac \gamma{12}+{\frac {7}{
36}} \right) \beta ^{2}-\left( \frac \gamma{12}+{\frac {
29}{144}} \right)  \beta ^{3} +O\bigl(\beta^4\bigr)
\end{aligned}
$$
obtained by means of the representation
$$
\zeta_B'(0;a,1,1)=\frac 1 {12}\left(a+\frac 1 a \right)\gamma-\frac 1 {12}\left(\frac 1 a +3+a\right)
\log a+\frac 5 {24}a-\frac 1 4 \log (2\pi)+ J(a).$$
Here 
\begin{equation}\label{Ja}
J(a)=\int_0^\infty\frac 1 {e^x-1}\left[\frac 1 {2x} {\coth\frac x{2a}}-\frac a  4 {\csch^2\frac x 2}-\frac 1 {12}\left(a+\frac1 a \right)\right]\,dx,
\end{equation}
$$
J'(a)=-\frac 1 {36}(a-1)+\frac 1 {16}(a-1)^2+O\bigl((a-1)^3\bigr);
$$
see~\eqref{Z'_a(0)} and~\cite[formulas (54) and (85)--(88)]{Au-Sal}. 

The determinant is unbounded from above due to~\eqref{ASE}.
\end{proof}
We end this subsection with remark that (in contrast to the result in Proposition~\ref{extremal}) the determinant of Friederichs Laplacian on the constant curvature spheres with two conical singularities and normalized curvature does not attain a local maximum on the metric of standard (smooth) round sphere. The corresponding asymptotic  for the right hand side of~\eqref{zetaPrimeSphere} as $\beta\to0$ can be easily obtained in the same way as in Proposition~\ref{extremal}, we leave it to the reader.

 \begin{figure} \centering\includegraphics[scale=.25]{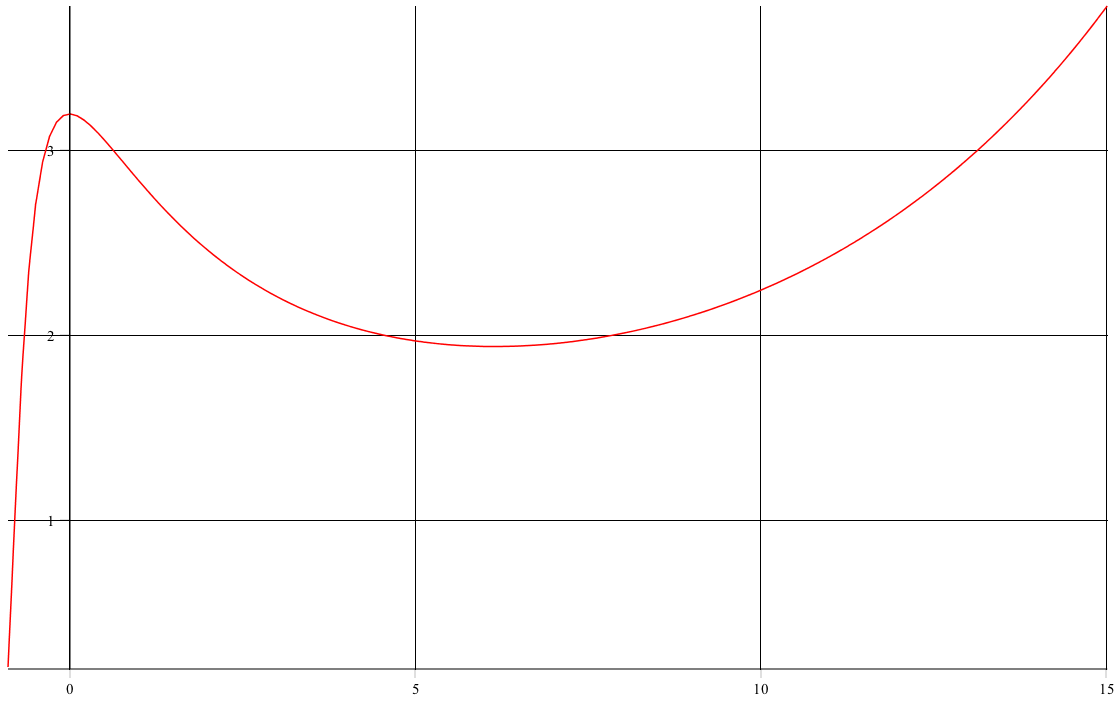}
  \caption{Metrics of constant curvature on a  sphere  with two antipodal conical singularities of angle $2\pi(\beta+1)$ and normalized area $4\pi$: a graph of $\det \Delta^{\varphi}_{\text{Area } 4\pi}$ as a function of $\beta>-0.89$ when $\mu=0$.}
  \label{F1}
\end{figure}

\subsection{Flat metrics}\label{FM}

 Consider the Riemann sphere $\Bbb CP^1$ with flat (curvature zero) metric 
\begin{equation}\label{flat}
m_\varphi=\prod_{j=1}^n|z-p_j|^{2\beta_j}|dz|^2
\end{equation} 
with $n\geq 3$ distinct conical singularities $p_j\in\Bbb C$ of order $\beta_j>-1$, $\sum \beta_j=|\pmb\beta|=-2$.
 \begin{proposition} For the determinant (of the Friederichs selfadjoint  extension) of the Laplacian $\Delta^\varphi$ on the Riemann sphere $\Bbb CP^1$ with flat conical  metric~\eqref{flat} we have
\begin{equation}\label{A-S4-}
 \begin{aligned}
\log\frac{\det \Delta^\varphi}{A_\varphi}
=&\frac  1 6 \sum_{j=1}^n\sum_{i=1, i\neq j}^n\frac{\beta_i \beta_j}{\beta_j+1}\log|p_i-p_j|
\\
&-\sum_{j=1}^nC(\beta_j)  -4\zeta_R'(-1)-\frac 4 3 \log 2+\frac 1 6-\log \pi,
\end{aligned}
\end{equation}
where $A_\varphi=\int_{\Bbb C}\prod_{j=1}^n|z-p_j|^{2\beta_j}\frac{dz\wedge d\bar z}{-2i}<\infty$ is the total area of $(\Bbb CP^1,m_\varphi)$. 
\end{proposition}

\begin{proof}  For the metric~\eqref{flat} we have $m_\varphi=e^{2\chi}|dz|^2$, where
$
\chi(z)=\sum\beta_j\log|z-p_j|$.
 As the reference metric  $m_0=e^{2\psi}|dz|^2$ we take the standard curvature one metric on $\Bbb CP^1$, i.e. $\psi(z)=\log2-\log(1+|z|^2)$. Thus  $A_0=4\pi$, $K_0=1$, and  $\det \Delta^0$ is given by~\eqref{WeisSphere}. 
Clearly, $K_\varphi=0$, $\varphi=\chi-\psi$, and the formula~\eqref{Pol} in Theorem~\ref{main} takes the form
 \begin{equation}\label{A-S1}
 \begin{aligned}
\log(\det &\Delta^\varphi/A_\varphi)+4\zeta_R'(-1)-1/2+\log(4\pi)
=-\frac{1}{12\pi}\int_{\Bbb C} \bigl(\chi-\psi\bigr) e^{2\psi} \frac{dz\wedge d\bar z}{-2i} 
\\ &+\frac  1 6 {\sum_{j=1}^n\beta_j}\left(\frac{1}{\beta_j+1}\sum_{i=1,i\neq j}^n\beta_i\log|p_i-p_j|- \log\frac 2{1+|p_j|^2}\right)-\sum_{j=1}^nC(\beta_j).
\end{aligned}
\end{equation}

 We remove the parentheses in the integral and evaluate the first term: 
 \begin{equation}\label{A-S3}
\begin{aligned}
\frac{1}{12\pi}\int_{\Bbb C}\chi e^{2\psi } \frac{dz\wedge d\bar z}{-2i}=\frac{1}{12\pi}\int_{\Bbb C}\chi (-4\partial_z\partial_{\bar z}\psi) \frac{dz\wedge d\bar z}{-2i}
\\
={\frac{1}{12\pi}\lim_{\epsilon\to0}\left(\oint_{|z|=1/\epsilon}+\sum_{j=1}^n\oint_{|z-p_j|=\epsilon}\right)\Bigl(\chi \partial_{\vec n}\psi-\psi \partial_{\vec n}\chi \Bigr)|dz|}
\\
=\frac{1}{12\pi}\sum_{j=1}^n\beta_j\lim_{\epsilon\to0}\left(\int_0^{2\pi}\Bigl(\epsilon\log 2 +o(\epsilon)\Bigr)\frac 1 \epsilon\,d\theta-\int_0^{2\pi}\Bigl(\psi(p_j)\frac 1 \epsilon +O(1)\Bigr)\epsilon d\theta\right)
\\
=-\frac 1 3 \log 2 -\frac{1}{6}\sum_{j=1}^n\beta_j  \log\frac 2{1+|p_j|^2}.
\end{aligned}
\end{equation}
Now~\eqref{A-S1},~\eqref{A-S3}, and~\eqref{intformula2}  imply~\eqref{A-S4-}. 
\end{proof}
\begin{remark} The formula~\eqref{Cbeta} for $C(\beta)$ together with  the identity $\sum\beta_j=-2$ allows to write~\eqref{A-S4-} in the form
\begin{equation}\label{A-S4}
 \begin{aligned}
\log\frac{\det \Delta^\varphi}{A_\varphi}
=&\frac  1 6 \sum_{j=1}^n\sum_{i=1, i\neq j}^n\frac{\beta_i \beta_j}{\beta_j+1}\log|p_i-p_j|
\\
&-\sum_{j=1}^n\left(2Z'_{\beta_j+1}(0)+\frac 1 2 \log(\beta_j+1)\right)- \log 2,
\end{aligned}
\end{equation}
where $Z'_{\beta+1}(0)$ is given by 
\begin{equation}\label{Z'_a(0)}
Z'_{\beta+1}(0)=\zeta'_B(0;\beta+1,1,1)-(\beta+1)\zeta'_R(-1)+\frac 1 {12} \Bigl(\beta+1 -\frac{1}{\beta+1}\Bigr) \log 2-\frac {\beta} 4 \log 2\pi.
\end{equation}
We note that~\eqref{A-S4} coincides with the celebrated Aurell-Salomonson formula~\cite[formula (50), where $\beta_j$ (resp. $p_j$) is denoted by  $-\beta_j$ (resp. $w_j$), and $\text{Area}=A({\cal M})=A_\varphi$]{Au-Sal 2}. One of equivalent definitions for $Z'_{\beta+1}(0)$ in~\cite{Au-Sal,Au-Sal 2} reads
\begin{equation}\label{ZIR}
\begin{aligned}
Z'_a(0)=\frac 1 {12}\left(\frac 1 a -a \right)(\gamma -&\log 2)-\frac 1 {12}\left(\frac 1 a +3+a\right)\log a +J(a)
\\
&-a\left(-\frac 1 6 \gamma-\frac 5 {24}+\frac 1 4 \log(2\pi)+\zeta'_R(-1)\right),
\end{aligned}
\end{equation}
where $\gamma=-\Gamma'(1)$ and $J(a)$ is the same as in~\eqref{Ja}; see~\cite[formulas (53), (54) and (85)--(87)]{Au-Sal}. One can easily check that for all rational numbers $a$ the values of $Z'_{a}(0)$ defined by~\eqref{Z'_a(0)} and~\eqref{ZIR} coincide  (we use Lemma~\ref{Barnes} in Appendix~\ref{BarnesPV} to evaluate~\eqref{Z'_a(0)}; for the evaluation of~\eqref{ZIR} we refer to~\cite[formula (102)]{Au-Sal}). Thus due to analytic regularity of $\Bbb R_+\ni a\mapsto Z'_a(0)$ the definitions~\eqref{Z'_a(0)} and~\eqref{ZIR} are equivalent. To the best of our knowledge, this is the first rigid mathematical proof of the Aurell-Salomonson formula.

\end{remark}

It is interesting to note that in~\cite{Khuri2} it was shown that the determinant $\det \Delta^\varphi$ is real analytic on the moduli space of suitably restricted metrics~\eqref{flat}.  Let us also note that if $m_0$ is the smooth hyperbolic (curvature $K_0=-1$) metric on a surface $M$ (of genus greater than one)  and $m_\varphi$ is a conformally equivalent flat  metric with conical singularities, then  Theorem~\ref{main}.1 returns the result of~\cite[Cor. 6.2]{McPk}. We end this subsection  with remark  that a formula for the derivative of the determinant of the Friederichs Dirichlet Laplacian on a flat circular convex sector in $\Bbb R^2$ with respect to the opening angle was recently found in~\cite{Clara-Rowlett}, though the results and methods of the latter paper seem not to be directly related to those presented here. 
   
\subsection{Constant curvature metric disks}\label{CMD}
Consider the unit disk $|z|\leq 1$ endowed with the metric
\begin{equation}\label{DilAn}
m_\varphi=|z|^{2\beta}e^{2\phi}|dz|^2
, \quad \phi(z)=\log 2 -\log(1+K|z|^{2\beta+2}),
\end{equation}
where  $K>-1$, $K_\varphi=(\beta+1)^2K$ is the curvature, and $2\pi(\beta+1)> 0 $ is the angle of conical singularity at $z=0$. 

\begin{proposition} Let $K>-1$ and $\beta>-1$. Then for the determinant  of (the Friederichs extension of) the Dirichlet Laplacian $\Delta^\varphi$ on the disk $|z|\leq 1$ with metric~\eqref{DilAn} we have 
\begin{equation}\label{CCD}
\begin{aligned}
\log\det\Delta^\varphi  = - 2\zeta'_B(0;\beta+1,1,1)-&  \frac 1 2 \log(\beta+1)\\ & +\frac{11K-5}{12(1+K)}(\beta+1)-\frac 1 2 \log(2\pi).
 \end{aligned}
 \end{equation}
\end{proposition}

In the curvature zero case  this result was obtained in~\cite{Weisberger} (if there is no  conical singularity, i.e. $K=\beta=0$) and in~\cite{Spreafico} (if there is a conical singularity at $z=0$, i.e. $\beta>-1$ and $K=0$), cf.~\eqref{zeta at zero+} and~\eqref{Spr+}.  We also note that in the case $K=1$ and $\beta=0$ the metric disk is isometric to the unit hemisphere and the formulas~\eqref{zeta_at_zero},~\eqref{CCD} return the corresponding result, cf.~\cite[formulas (24)]{Weisberger}. In all other cases the result is new. 

\begin{proof}
 Let us take the flat metric $|dz|^2$ as the reference metric $m_0$ and do the calculations for $m_\varphi=|z|^{2\beta}e^{2\phi}|dz|^2$ with $$\phi(z)= -\log(1+K|z|^{2\beta+2}).$$ Then one can use the standard rescaling argument in order to add $\log 2$ to $\phi$ and obtain~\eqref{CCD} for the determinant corresponding to the metric in~\eqref{DilAn}; this will only decrease $\log\det\Delta^\varphi=-\zeta'(0)$ by $\zeta(0)2\log 2$, where 
$
\zeta(0)  =\frac{1}{12}\left(\beta+1+\frac 1 {\beta+1}\right)
$, 
cf.~\eqref{zeta_at_zero}.
   
    From the formula~\eqref{Pol-Alv} in Theorem~\ref{main} we get
\begin{equation*}
\begin{aligned}
&-\zeta'(0)-({1}/{3})\log 2+(1/2) \log(2\pi)+{5}/{12}+2\zeta_R'(-1)\\&=-\frac{1}{12\pi}\left(\int_{|z|\leq 1}K_\varphi \varphi |z|^{2\beta}e^{2\phi} \frac{dz\wedge d\bar z}{-2i}+\int_{|z|=1}\varphi\partial_{|z|}\varphi\,ds_0\right)\\&-\frac 1{6\pi}\int_{|z|=1}k_0\varphi\, ds_0
-\frac 1 {4\pi} \int_{|z|=1}\partial_{|z|}\varphi \,ds_0-C(\beta),
\end{aligned}
\end{equation*}
 where
$\varphi(z)=\beta\log|z|+\phi(z)$ and $K_\varphi=(2\beta+2)^2K$; see~\eqref{E1} for the value of $ \log \det \Delta^{0}$.  We have
$$
\begin{aligned}
\int_{|z|\leq 1}K_\varphi\varphi |z|^{2\beta}e^{2\phi} \frac{dz\wedge d\bar z}{-2i}=\beta\int_{|z|\leq 1}\log|z|(-&4\partial_z\partial_{\bar z}\phi)\frac{dz\wedge d\bar z}{-2i} 
\\& +(2\beta+2)^2K\int_{|z|\leq 1}\phi |z|^{2\beta}e^{2\phi} \frac{dz\wedge d\bar z}{-2i}.
\end{aligned}
$$
Here
$$
\int_{|z|\leq 1}\log|z|(-4\partial_z\partial_{\bar z}\phi)\frac{dz\wedge d\bar z}{-2i} =\oint_{|z|=1}\frac{\phi(z)}{|z|}|dz|-\lim_{\epsilon\to0+}\oint_{|z|=\epsilon}\frac{\phi(z)}{|z|}|dz|
=-2\pi\log(1+K)
$$
and
$$\begin{aligned}
\int_{|z|\leq 1}\phi |z|^{2\beta}e^{2\phi} \frac{dz\wedge d\bar z}{-2i}=\frac{-\pi}{\beta+1}\int_0^1(1+K u)^{-2}\log(1+Ku)\,du\\=\frac \pi{\beta+1} \frac{\log(1 + K)-K}{K(1+K)}.
\end{aligned}$$

Next we evaluate the integrals along the circle $|z|=1$ and get
$$
\int_{|z|=1}\varphi\partial_{|z|}\varphi\,ds_0={2\pi}\log(1+K)\left(\frac{(2\beta+2)K}{1+K}-\beta\right),\quad
\int_{|z|=1}\varphi\, ds_0=-2\pi  \log(1+K),
$$
$$
 \int_{|z|=1}\partial_{|z|}\varphi \,ds_0={2\pi}\left(\beta-\frac{(2\beta+2)K}{1+K}\right). 
 $$
 These calculations together with formula~\eqref{Cbeta} for  $C(\beta)$ and the rescaling mentioned in the beginning of the proof imply~\eqref{CCD}.
  \end{proof}

\subsection{Hyperbolic spheres}\label{HS}
As is known, there exists a unique hyperbolic  (curvature $K_\varphi=-1$) conformal metric  $m_\varphi=e^{2\varphi}|dz|^2$ on the Riemann sphere $\Bbb CP^1$ with conical singularities of order $\beta_j>-1$ at $p_j\in \Bbb CP^1$,  $j=1,\dots,n$, provided $n\geq 3$ and $\sum_{j=1}^n\beta_j=|\pmb\beta|<-2$; see e.g.~\cite{Troyanov Curv}.  We shall assume that $p_n=\infty$ (and then $\beta_n=0$ and $n\geq 4$ if there is no conical singularity at infinity).
 The corresponding metric potential  $\varphi$  has  the following asymptotics in a neigborhood of each $p_j$:
\begin{equation}\label{AE}
\begin{aligned}
\varphi(z) & =\beta_j\log|z-p_j|+\phi_j+O(|z-p_j|^{2\beta_j+2}),\quad z\to p_j,\quad 0<j<n,\\
\varphi(z) & =-(\beta_n+2)\log|z|+\phi_n+O(|z|^{-2\beta_n-2}),\quad z\to\infty;
\end{aligned}
\end{equation}
the asymptotics can be differentiated.  We first express the determinant of  Laplacian on $(\Bbb CP^1,e^{2\varphi}|dz|^2)$ in terms of  the metric potential $\varphi$. 

\begin{proposition}\label{PHS} For the spectral determinant $\det \Delta^\varphi$ of the Friederichs extension of Laplacian $\Delta^\varphi$ on the hyperbolic (curvature $K_\varphi=-1$) sphere $(\Bbb CP^1,e^{2\varphi}|dz|^2)$ we have 
$$\begin{aligned}
\log \det \Delta^\varphi=\log \left(-2-|\pmb\beta|\right)+\frac{1}{12\pi}\int_{\Bbb C} \varphi \,e^{2\varphi}\, \frac{dz\wedge d\bar z}{-2i} -\frac 1 6 \left(1+\frac 1 {\beta_n+1}\right)\phi_n\\
+\frac  1 6 \sum_{j=1}^{n-1}\frac {\beta_j}{\beta_j+1}\phi_j-\sum_{j=1}^{n}C(\beta_j)  -\frac 1 3 \log 2 +\frac 1 {6} -4\zeta'_R(-1),
\end{aligned}
$$
where  $\phi_j$ stands for the constant term in the asymptotic expansion~\eqref{AE}, and $C(\beta)$ is defined in~\eqref{Cbeta}.  
\end{proposition}

\begin{proof} By using the BFK formula we cut the Riemann sphere  into two pieces along the circle $|z|=1/\epsilon$ as $\epsilon\to0+$. For the Neumann jump operator on this circle we have 
$$ 
\frac {\det R^\varphi_{|z|=1/\epsilon}}{L_\varphi(|z|=1/\epsilon)}=\frac 1 2;
$$ 
this can be easily seen from the conformal invariance of the left hand side (see Lemma~\ref{confin}) together with BFK formula~\eqref{Fr08:50}  and  formulas for the determinant of Laplacian on the unit sphere  and hemisphere (see~\cite{Weisberger} or~\eqref{WeisSphere} and \eqref{CCD} with $K=1$ and $\beta=0$).
The Gauss-Bonnet theorem~\cite{Troyanov Curv} implies that the total area $A_\varphi$ of hyperbolic sphere $(\Bbb CP^1, e^{2\varphi}|dz|^2)$ is  $A_\varphi=-2\pi\chi(\Bbb CP^1,\pmb \beta)=-2\pi(2+|\pmb\beta|)$. Thus we have 
\begin{equation}\label{ee1}
\det\Delta^\varphi=-\pi\Bigl(2+|\pmb\beta|\Bigr)\lim_{\epsilon\to0+}\left(\det\Delta^\varphi_{|z|\leq 1/\epsilon}\det\Delta^\varphi_{|z|\geq 1/\epsilon}\right),
\end{equation}
where
\begin{equation}\label{ee2}\begin{aligned}
\log\det\Delta^\varphi_{|z|\geq1/\epsilon}=&-\frac 1 {6}\left((\beta_n+1)^2+1\right)\log \epsilon -\frac 1 6 \left(\beta_n+1+\frac 1 {\beta_n+1}\right)\phi_n
\\
&-C(\beta_n)-\zeta'(0;1)-\frac {\beta_n} 2+o(1).
\end{aligned}
\end{equation}
The last formula for the determinant of Dirichlet Laplacian in the disk  $|w|\leq \epsilon$, $w=1/z$, 
immediately follows from~\eqref{Pol-Alv},~\eqref{AE},  and~\eqref{E1} (it can also be  obtained from~\eqref{CCD}  if in the local parameter $w$ the metric $e^{2\varphi}|dz|^2$ takes the form~\eqref{DilAn}).

Let $m_0=|dz|^2$. Then thanks to the Polyakov-Alvarez type formula~\eqref{Pol-Alv}  we get
\begin{equation}\label{ee3}\begin{aligned}
&\log\frac{\det\Delta^\varphi_{|z|\leq 1/\epsilon}}{\det\Delta^0_{|z|\leq 1/\epsilon}}=\frac{1}{12\pi}\int_{|z|\leq 1/\epsilon} \varphi\,dA_\varphi -\frac{1}{12\pi}\int_{|z|=1/\epsilon} \varphi\partial_{\vec n}\varphi\,ds_0 \\&-\frac 1{6\pi}\int_{|z|=1/\epsilon}\epsilon\varphi\, ds_0
-\frac 1 {4\pi}\int_{|z|=1/\epsilon} \partial_{\vec n}\varphi\,ds_0 +\frac  1 6 \sum_{j=1}^{n-1}\frac {\beta_j}{\beta_j+1}\phi_j-\sum_{j=1}^{n-1}C(\beta_j),
\end{aligned}
\end{equation}
where  $$
\log  \det\Delta^0_{|z|\leq 1/\epsilon}=\frac 1  3 \log\epsilon +\frac 1 3 \log 2-\frac 1 2 \log 2\pi -\frac 5 {12} -2\zeta'_R(-1);$$ cf.~\eqref{E1}. 

It is easy to verify that 
$$
-\frac 1 {4\pi}\int_{|z|=1/\epsilon}\partial_{\vec n}\varphi\, ds_0= \frac{\beta_n+2} 2+ o(1),
$$

$$
-\frac 1 {6\pi}\int_{|z|=1/\epsilon}\epsilon\varphi \,ds_0=
 -\frac {\beta_n+2} 3 \log \epsilon -\frac 1 3 \phi_n +o(1),
$$
$$
-\frac 1 {12\pi}\int_{|z|=1/\epsilon}\varphi\partial_{\vec n}\varphi\, ds_0=
\frac{1} 6\Bigl( (\beta_n+2)^2\log\epsilon+(\beta_n+2)\phi_n\Bigr)+o(1).
$$
 This together with~\eqref{ee1}--\eqref{ee3} completes the proof. 
\end{proof}

Consider the mapping 
\begin{equation}\label{covering}
w=f(z)=\frac {z^{2}}{1+\mu z^{2}}: \Bbb CP^1\to \Bbb CP^1
\end{equation}
 with  $\mu\in[0,\infty)$. It is a ramified covering with ramification divisor $1\cdot 0 + 1\cdot \infty$.  The pull back  of $m_\varphi$ by $f$ is a hyperbolic (curvature $-1$) metric on $\Bbb CP^1$ with  potential  
$f^*\varphi=\varphi\circ f+\log|f'|$. For brevity we assume that $0$ and $1/\mu$ are not among the conical singularities of $m_\varphi$ (for any  $\mu\in[0,\infty)$). Then $ e^{2f^*\varphi}|dz|^2$  has conical singularities of order $\beta_j$ at the pre-images $z=\pm\sqrt{\frac {p_j}{1-\mu p_j}}$  of the conical singularities $p_1,\dots, p_n$ as well as conical singularities of order $1$ at $z=0$ and $z=\infty$. By setting $\rho=e^{2\varphi}$, $z_1=0$, and $z_2=1/\mu$ in~\cite[formula (1.2)]{KalvinJGA} one obtains the variational formula
\begin{equation}\label{Var}
 \det \Delta^{f^*\varphi}=C\mu^{-1/2}\sqrt[8]{e^{2\varphi(0)}e^{2\varphi(1/\mu)}},
\end{equation}
where $\Delta^{f^*\varphi}$ is the Friederichs extension of Laplacian on $(\Bbb CP^1,e^{2f^*\varphi}|dz|^2)$ and $C$ is an undetermined constant that does not depend on the parameter $\mu\in[0,\infty)$. In Proposition~\ref{00:15} below we independently deduce~\eqref{Var} and find that 
\begin{equation}\label{VarC}
C=-\frac {2^{2/3}e^{6\zeta_R'(-1)}}{ 2+|\pmb\beta|} ( \det\Delta^\varphi)^2,
\end{equation}
where the divisor $\pmb\beta$ and $\det\Delta^\varphi$ are the same as in Proposition~\ref{PHS}.
\begin{proposition}\label{00:15}  Let $f^*\varphi$ stand for the potential of the pull back of the  hyperbolic metric $e^{2\varphi(z)}|dz|^2$ by the mapping~\eqref{covering}. Then the spectral determinant of the Friederichs extension of Laplacian $\Delta^{f^*\varphi}$ on $(\Bbb CP^1, e^{2f^*\varphi}|dz|^2)$ satisfies~\eqref{Var} with $C$ specified in~\eqref{VarC}.
\end{proposition}

\begin{proof} We will rely on Proposition~\ref{PHS} with $\varphi$ replaced by $f^*\varphi$. Notice that for the corresponding integral we have
\begin{equation}\label{1040int}
\begin{aligned}
\int_{\Bbb C} f^*\varphi \,e^{2f^*\varphi}\, \frac{dz\wedge d\bar z}{-2i}
=\int_{\Bbb C} (\log|f'|)  \,e^{2f^*\varphi}\, \frac{d z\wedge d\bar  z}{-2i}+2\int_{\Bbb C} \varphi \,e^{2\varphi}\, \frac{dz\wedge d\bar z}{-2i},
\end{aligned}
\end{equation}
where the last integral can be expressed in terms of $\det\Delta^\varphi$. Let $\Bbb C^\epsilon$ stand for the annulus $\{z\in\Bbb C:\epsilon\leq|z|\leq1/\epsilon\}$ minus the union of the epsilon neighborhoods of all pre-images  $z=\pm\sqrt{\frac {p_j}{1-\mu p_j}}$  of conical singularities $p_1,\dots, p_n$.  Then the Liouville equation together with  Stokes' theorem gives
\begin{equation}\label{1519Wed}
\begin{aligned}
\int_{\Bbb C} (\log|f'| ) \,e^{2f^*\varphi}\, \frac{dz\wedge d\bar z}{-2i}=-\lim_{\epsilon\to 0+} \int_{\Bbb C^\epsilon}  (\log|f'| )\bigl(-4\partial_z\partial_{\bar z}( \varphi\circ f)\bigr) \, \frac{dz\wedge d\bar z}{-2i}\\
=-\lim_{\epsilon\to0+}\oint_{\partial \Bbb C^\epsilon}
\Bigl[(\varphi\circ f)\partial_{\vec n}   (\log|f'| )-  (\log|f'| )\partial_{\vec n} (\varphi\circ f)\Bigr ]|dz|,
\end{aligned}
\end{equation}
where the right hand side can be easily evaluated based on~\eqref{AE} and the explicit expression for $f$  (it suffices to take into account only logarithmic and constant terms of asymptotics for $\varphi\circ f$ and $\log|f'|$ as $z$ approaches a conical singularity of $e^{2f^*\varphi}|dz|^2$; we omit the details). This together with $ f^*\varphi=\varphi\circ f+\log|f'|$ allows to write the formula from Proposition~\ref{PHS} in the form 
$$
\begin{aligned}
\log \det\Delta^{f^*\varphi}=\log (-2-|\pmb\beta|)+\frac{1}{6\pi}\int_{\Bbb C} \varphi \,e^{2\varphi}\, \frac{dz\wedge d\bar z}{-2i}-\frac 1 3\left(1+\frac{1}{\beta_n+1} \right) \phi_n 
\\
+\frac 1 3 \sum_{j=1}^{n-1}\frac {\beta_j}{\beta_j+1}\phi_j -2\sum_{j=1}^{n}C(\beta_j)-2C(1)  -\frac 1 6  \log 2 +\frac 1 {6} -4\zeta'_R(-1)
\\
-\frac 1 2 \log \mu +\frac 1 4\bigl(\varphi(0)+ \varphi(1/\mu)\bigr),
\end{aligned}
$$
where $C(1)=-\zeta_R'(-1)-\frac 1 {12}\log 2 -\frac 1 {12}$ (see~\eqref{Cbeta} and~Lemma~\ref{Barnes}). This together with Proposition~\ref{PHS} completes the proof. 
\end{proof}

\begin{remark} It is also interesting to compare our results to those in the work~\cite{C-W} about quantum Hall states on singular surfaces. With this aim in mind we introduce the  Liouville functional $\mathcal S$ by the formula
$$
\mathcal S=-\int_{\Bbb C} \varphi \,e^{2\varphi}\, \frac{dz\wedge d\bar z}{-2i}+2\pi \sum_{j=1}^{n-1}\beta_j\phi_j+2\pi(\beta_n+2)\phi_n
$$
and denote
$$
\mathcal H=\exp\left\{\sum_{j=1}^{n-1}2\beta_j\left(1+\frac{1}{\beta_j+1}\right)\phi_j\right\}.
$$
   Recall that the function $\mathcal S$ on the moduli space $\mathcal M_{0,n}$ serves as a K\"ahler potential for a family of  K\"ahler metrics  (parametrized by the orders $\beta_j\in(-1,0)$ of conical singularities) and   generates  the famous accessory parameters, see~\cite{TZ1} (where the Liouville functional is defined as  $\mathcal S-2\pi(2+|\pmb \beta|)$; here the  moduli independent term $-2\pi(2+|\pmb \beta|)$  is the total area of the singular hyperbolic sphere $(\Bbb CP^1, e^{2\varphi}|dz|^2)$). 
   Now the result of Proposition~\ref{PHS} can be equivalently rewritten in the form  
\begin{equation*}\label{LAMS}
\begin{aligned}
\log \det \Delta^\varphi  = &\log(-2-|\pmb\beta|) -\frac 1 {12\pi}(\mathcal S-  \pi\log \mathcal H )
\\ &+ \frac 1 6 \left(\beta_n+1-\frac{1}{\beta_n+1}\right)\phi_n-\sum_{j=1}^n C(\beta_j)-\frac 1 3 \log 2 - 4\zeta_R'(-1).
\end{aligned}
\end{equation*}
We note that up to the moduli independent terms, i.e. in the form $$\log \det \Delta^\varphi  \varpropto -\frac 1 {12\pi}(\mathcal S-  \pi\log \mathcal H ),$$ this equation was recently
 established (at the physical level of rigor) in~\cite{C-W}.  Let us also recall that although in the case of a hyperbolic sphere with only elliptic singularities of finite order  (or, equivalently, in the case  $\beta_j+1=1/q_j$ with positive integer $q_j$ for all  $j=1,\dots,n$, cf.~\eqref{BLF2}) there exists a ramified covering of $\Bbb CP^1$ by the upper hyperbolic half-plane $\Bbb H$ branched over $p_1,\dots,p_n$ and the determinant of Laplacian can  be introduced and studied in terms of the 
 Selberg zeta function (at least up to the moduli independent terms, cf.~\cite{TZ2,Teo} ), 
 this approach fails for conical singularities of general type we consider here.   
\end{remark}

\subsection{Genus $g> 1$ singular surfaces without boundary}\label{KOKOT}

Here we present a general formula for the determinant of Friederichs Laplacian on genus $g>1$ Riemann surface $M$ without boundary.  The result is based on the formula for the determinant  of Laplacians on compact polyhedral surfaces with trivial holonomy~\cite{KokotKorot} and Corollary~\ref{2conical}. This is a straightforward generalization of the scheme in~\cite{ProcAMS}:   Corollary~\ref{2conical}.1 together with the particular values  $C(0)=0$ and $C(1)=-\zeta_R'(-1)-\frac 1 {12}(\log 2+1)$ of the function $C(\beta)$ in~\eqref{Cbeta} should be used instead of~\cite[Prop.1]{ProcAMS}. Therefore we only formulate the result and omit the proof. In the particular case of a flat  metric $m_\varphi$, Proposition~\ref{Kokot}  with  $K_\varphi=0$  below is  essentially a  reformulation of  the main result in~\cite{ProcAMS} with the unnecessary additional assumption $\supp\pmb\alpha\cap\supp\pmb\beta=\varnothing $  removed and the invariance of the result in Corollary~\ref{2conical}.1 utilized.

\begin{proposition}\label{Kokot} Let $\omega$ be a holomorphic one-form on $M$ with $2g-2$ simple zeros and let  $m_0=|\omega|^2$ be the corresponding flat  metric (with trivial holonomy and first order conical singularities at the zeros of $\omega$). Consider  an admissible   metric  $m_\varphi=e^{2\varphi}m_0$ of  curvature $K_\varphi$. 
Let $\{P_1,\dots P_n\}$ be the set of all distinct points in the union $\supp\pmb\alpha\cup\supp\pmb\beta $, where  $\pmb\beta$ (resp. $\pmb\alpha$) is the divisor of $m_\varphi$ (resp. of  $m_0$ and $\omega$). 
Pick   a local holomorphic parameter $x$ centred at $P^j$ such that $m_0=|x|^{2\alpha_j}|dx|^2$ and  $m_\varphi=|x|^{2\beta_j}e^{2\phi_j}|dx|^2$, where  $\alpha_j=0$ if  $P^j\notin\supp\pmb\alpha$, $\alpha_j=1$ if $P^j\in\supp\pmb\alpha$, and $\beta_j=0$ if $P^j\notin\supp\pmb\beta$.   Then the determinant of the Friederichs Laplacian $\Delta^\varphi$ on $(M, m_\varphi)$ satisfies
\begin{equation*}
\begin{aligned}
\det \Delta^\varphi=(2\pi)^{-4/3}\kappa_0^{g-1}A_\varphi(\det\Im\Bbb B)|\tau_g(M,\omega)|^2\exp\Biggl\{-\frac{1}{12\pi}\int_{M} K_\varphi \varphi\,dA_\varphi
\\
 + \sum_{j=1}^n\Biggl ( \frac  {\phi_j(0)} {6} \Bigl (\frac{\beta_j}{\beta_j+1}+\alpha_j \Bigr)-C(\beta_j)\Biggr) -(2g-2)\Bigl(\zeta_R'(-1)+\frac 1 {12}(\log 2 +1)\Bigr)\Biggr\},
\end{aligned}
\end{equation*}
where $\kappa_0$ is an absolute constant that can be expressed in terms of spectral determinants of some model operators, $A_\varphi$ is the total area of $M$ in the metric $m_\varphi$,  $\Bbb B$ is the matrix of $b$-periods of the Riemann surface $M$, and the Bergman tau-function $\tau$ is a holomorphic function that admits explicit expression through theta-functions, prime forms, and the divisor $\pmb \alpha$.
\end{proposition}

\appendix
\section{Derivative $\zeta_B'(0;\beta+1,1,1)$ of the Barnes zeta function  for rational values of $\beta$}\label{BarnesPV}

\begin{lemma}\label{Barnes} Let $p$ and $q$ be coprime natural numbers. Then  
\begin{equation}\label{BZetaPrime}
 \begin{aligned}
\zeta'_B(0;p/q,1,1)=&\frac 1 {pq}\zeta_R'(-1)-\frac{1}{12pq}\log q
+\left(\frac 1 4 {+}S(q,p)\right)\log{\frac{q}{p}}
\\
+\sum_{k=1}^{p-1}\left(\frac 1 2 -\frac k p \right)&\log \Gamma\left(  \left(\!\!\!\left(\frac{kq}{p}\right)\!\!\!\right)+\frac 1 2 \right)+\sum_{j=1}^{q-1}\left(\frac 1 2 -\frac j q \right)\log \Gamma\left(\left(\!\!\!\left(\frac{jp}{q}\right)\!\!\!\right)+\frac 1 2\right),
\end{aligned}
\end{equation}
where $S(q,p)=\sum_{j=1}^{p}\left(\!\!\left(\frac{j}{p}\right)\!\!\right) \left(\!\!\left(\frac{jq}{p}\right)\!\!\right)$ is the Dedekind sum, and the symbol $(\!(\cdot)\!)$ is defined so that   $(\!(x)\!)=x-\lfloor x\rfloor-1/2$ for $x$ not an integer and $(\!(x)\!)=0$ for $x$ an integer.

In particular, for natural numbers  $p$ and $q$ one has \begin{equation}\label{BLF1}
\zeta_B'(0;p,1,1)=\frac 1 p \zeta'_R(-1)-\left(\frac{p}{12}+\frac 1 4 +\frac 1{6p}\right)\log p -\sum_{j=1}^{p-1}\frac j p \log\Gamma\left(\frac j p\right)+\frac{p-1}4 \log 2\pi,
\end{equation}
\begin{equation}\label{BLF2}
\zeta_B'\left(0; 1 /q,1,1\right)=\frac 1 q \zeta'_R(-1)-\frac 1 {12q}\log q-\sum_{j=1}^{q-1}\frac j q\log\Gamma\left(\frac j q\right)+\frac{q-1}4 \log 2\pi.
\end{equation}
Thus $\zeta_B'(0;1,1,1)=\zeta'_R(-1),$
$$
\zeta_B'(0;2,1,1)=\frac 1 2 \zeta'_R(-1)-\frac 1 4 \log 2,\quad \zeta_B'\left(0;1 /2,1,1\right)=\frac 1 2 \zeta'_R(-1)+\frac 5 {24}\log 2,
$$
$$
\zeta_B'(0;3,1,1)=\frac 1 3 \zeta'_R(-1)+\frac 1 6 \log 2-\frac 7 {18}\log 3 -\frac 1 3 \log\Gamma \left(\frac 2 3 \right)+\frac 1 6 \log \pi,
$$
$$ 
\zeta_B'\left(0;1 /3,1,1\right)=\frac 1 3 \zeta'_R(-1)+\frac 1 6 \log 2+\frac 5 {36}\log 3 -\frac 1 3 \log\Gamma \left(\frac 2 3 \right)+\frac 1 6 \log \pi,
$$
$$
\zeta_B'(0;4,1,1)=\frac 1 4 \zeta'_R(-1)-\frac 5 8 \log 2-\frac 1 2 \log \Gamma\left(\frac 3 4 \right)+\frac 1 4 \log \pi,$$
$$
\zeta_B'\left (0;1/ 4,1,1\right)=\frac 1 4 \zeta'_R(-1)+\frac 7 {12} \log 2-\frac 1 2 \log \Gamma\left(\frac 3 4 \right)+\frac 1 4 \log \pi,\dots$$

\end{lemma}
\begin{proof} Let us first prove~\eqref{BLF1}. Notice that
$$\zeta_B(s;1,1,1)=\sum_{m,n=0}^\infty(m+n+1)^{-s}=\sum_{\ell=0}^\infty\sum_{n=0}^\ell(\ell+1)^{-s}=\zeta_R(s-1),$$
where the left hand side can also be represented as the sum 
$\sum_{k=1}^p\zeta_B(s;p,1,k)$.  For each term of this sum we have
$$\begin{aligned}
&\zeta_B(s;p,1,k)
=\sum_{m=0}^\infty\sum_{n=k-1}^\infty(pm+n+1)^{-s}
\\
&=\sum_{m=0}^\infty\left(\sum_{n=0}^\infty(pm+n+1)^{-s}-\sum_{n=0}^{k-2}(pm+n+1)^{-s}\right)
=\zeta_B(s;p,1,1)-p^{-s}\sum_{j=1}^{k-1}\zeta_H(s;j/p),
\end{aligned}
$$
where
$$
\zeta_H(s; x )=\sum_{m=0}^\infty(m+x)^{-s}
$$
is the Hurwitz zeta function.  Solving the resulting equation for $\zeta_B(s;p,1,1)$ we obtain
$$\begin{aligned}
\zeta_B(s;p,1,1)=\frac 1 p \zeta_R(s-1)+p^{-s-1}\sum_{k=1}^p\sum_{j=1}^{k-1}\zeta_H(s;j/p)
\\=\frac 1 p \zeta_R(s-1)+p^{-s-1}\sum_{j=1}^{p-1}\sum_{k=j+1}^{p}\zeta_H(s;j/p)
\\=\frac 1 p \zeta_R(s-1)+p^{-s-1}\sum_{j=1}^{p-1}(p-j)\zeta_H(s;j/p).
\end{aligned}
$$
Now we differentiate with respect to $s$ and use the well-known identities
\begin{equation}\label{HurwitzZeta}
\zeta_H(0;x)=\frac 1 2 -x,\quad \zeta_H'(0;x)=\log\Gamma(x)-\frac 1 2 \log 2\pi.
\end{equation}
As a result we obtain
$$
\zeta_B'(0;p,1,1)=\frac 1 p \zeta'_R(-1)+\sum_{j=1}^{p-1}\left(1-\frac j p \right)\left(\log\Gamma\left(\frac j p\right)+\frac j p\log p\right)-\frac{p-1}4 (\log 2\pi+\log p).
$$
Taking into account that $$\sum_{j=1}^{p-1}\frac {j^2}{p^2}=\frac p 3 -\frac 1 2+\frac 1{6p}, \quad \prod_{j=1}^{p-1}\Gamma\left(\frac j p\right)=(2\pi)^{(p-1)/2}p^{-1/2}, $$   we arrive at~\eqref{BLF1}.

 In order to prove~\eqref{BLF2} we use the relation
$\zeta_B(s;1/q,1,1)=q^{s}\zeta_B(s;q,1,q)$ and obtain
$$
\begin{aligned}
\zeta_B(s;1/q,1,1)&=q^s\left(\zeta_B(s;q,1,1)-q^{-s}\sum_{j=1}^{q-1}\zeta_H(s;j/q)\right)
\\
&=q^{s-1}\zeta_R(s-1)-\sum_{j=1}^{q-1}\frac j q \zeta_H\left(s;\frac j q\right). 
\end{aligned}
$$
Since $\zeta_R(-1)=-1/12$, this implies~\eqref{BLF2}.

Let $p$ and $q$ be coprime.  B\'ezout's identity reads $xp+yq=1$. Without loss of generality  we can assume that $x\leq 0$ and $y\geq 0$ (otherwise take $x:=x-q$ and $y:=y+p$). 
We start with
$$
\zeta_B(s;p/q,1,1)=q^s\zeta_B(s;p,q,q). 
$$
We have 
$$\begin{aligned}
&\zeta_B(s;p,q,q+k)=\sum_{n=0}^\infty\left(\sum_{m=-xk}^\infty+\sum_{m=0}^{-xk-1}\right)(pm+qn+q+k)^{-s}
\\
&=\sum_{n=yk}^\infty\sum_{m=0}^\infty(pm+qn+q)^{-s}+q^{-s}\sum_{m=0}^{-xk-1}\zeta_H(s;(pm+q+k)/q)
\\
&=\zeta_B(s;p,q,q)-p^{-s}\sum_{n=0}^{yk-1}\zeta_H(s;(qn+q)/p)+q^{-s}\sum_{m=0}^{-xk-1}\zeta_H(s;(pm+q+k)/q). 
\end{aligned}
$$
Therefore
$$\begin{aligned}
\zeta_B(s;p,q,q)=\zeta_B(s;p,q,q+k)+p^{-s}\sum_{n=0}^{yk-1}\zeta_H(s;(qn+q)/p)-q^{-s}\sum_{m=0}^{-xk-1}\zeta_H(s;(pm+q+k)/q)
\\
=\frac{1}{p}\left(\zeta_B(s;1,q,q)+p^{-s}\sum_{k=0}^{p-1}\sum_{n=0}^{yk-1}\zeta_H(s;(qn+q)/p)-q^{-s}\sum_{k=0}^{p-1}\sum_{m=0}^{-xk-1}\zeta_H(s;(pm+q+k)/q)\right). 
\end{aligned}
$$
Here
$$\begin{aligned}
\zeta_B(s;1,q,q)&=\zeta_B(s;1,q,q+j)+q^{-s}\sum_{m=0}^{j-1}\zeta_H(s;(m+q)/q)
\\
&=\frac{1}{q}\left(\zeta_B(s;1,1,q)+q^{-s}\sum_{j=0}^{q-1}\sum_{m=0}^{j-1}\zeta_H(s;(m+q)/q)\right)
\\
&=\frac{1}{q}\left(\zeta_R(s-1)-\sum_{m=1}^{q-1}\zeta_H(s;m)+q^{-s}\sum_{j=0}^{q-1}\sum_{m=0}^{j-1}\zeta_H(s;(m+q)/q)\right).
\end{aligned}
$$

Finally we get 
$$
\begin{aligned}
\zeta_B(s;p/q,1,1)=\frac{q^s}{p q}\left(\zeta_R(s-1)-\sum_{m=1}^{q-1}\zeta_H(s;m)\right)+\frac 1 {pq }\sum_{j=1}^{q-1}\sum_{m=0}^{j-1}\zeta_H(s;(m+q)/q)
\\+q^sp^{-s-1}\sum_{k=1}^{p-1}\sum_{n=0}^{yk-1}\zeta_H(s;(qn+q)/p)-\frac 1 {p}\sum_{k=1}^{p-1}\sum_{m=0}^{-xk-1}\zeta_H(s;(pm+q+k)/q). 
\end{aligned}
$$
This together with~\eqref{HurwitzZeta} gives a representation of $\zeta_B'(0;p/q,1,1)$ in terms of $\zeta'_R(-1)$ and gamma functions. Then Gauss' multiplication formula allows to write the result in the form~\eqref{BZetaPrime}. A similar computation can be found in~\cite[Section 4]{Dowker}.
\end{proof}

\end{document}